\pdfoutput=1
\documentclass[11pt,letterpaper]{article}
\usepackage[lmargin=1.0in,rmargin=1.0in,bottom=1.0in,top=1.0in,twoside=False]{geometry}

\usepackage[utf8]{inputenc}
\usepackage[T1]{fontenc}
\usepackage{lmodern}

\usepackage{fullpage,amssymb,amsmath}
\usepackage{graphicx}%,algorithmic,algorithm}%, verbatim}
\usepackage{tikz}
\usetikzlibrary{shapes}

\usepackage{xcolor}
\usepackage{mathtools}
\usepackage{microtype}
\usepackage{amsfonts}
\usepackage{comment}
\usepackage[english]{babel}
\usepackage{mathrsfs}
\usepackage[ruled,vlined]{algorithm2e}
\usepackage{blindtext}
\usepackage{thmtools}
\usepackage{thm-restate}
\usepackage{booktabs}
\usepackage[sort&compress,numbers]{natbib}

\usepackage{trimspaces}
\usepackage{nccfoots}
\usepackage{setspace}
\usepackage{inconsolata}
\usepackage{libertine}
\usepackage[absolute]{textpos}

\usepackage[backgroundcolor = blue!50]{todonotes}
\usepackage{longtable}

\definecolor{blue}{rgb}{0.1,0.2,0.5}
\definecolor{brown}{rgb}{0.6,0.6,0.2}
\usepackage[ocgcolorlinks, linkcolor={blue}, citecolor={brown}]{hyperref}
\usepackage[amsmath,amsthm,thmmarks,hyperref]{ntheorem}
\usepackage{enumerate}
\usepackage{latexsym}
\usepackage{changepage}

\usepackage[nameinlink]{cleveref}

\crefformat{page}{#2page~#1#3}%
\Crefformat{page}{#2Page~#1#3}%
\crefformat{equation}{#2(#1)#3}%
\Crefformat{equation}{#2(#1)#3}%
\crefformat{figure}{#2Figure~#1#3}%
\Crefformat{figure}{#2Figure~#1#3}%
\crefformat{section}{#2Section~#1#3}
\Crefformat{section}{#2Section~#1#3}
\crefformat{chapter}{#2Chapter~#1#3}
\Crefformat{chapter}{#2Chapter~#1#3}
\crefformat{chapter*}{#2Chapter~#1#3}
\Crefformat{chapter*}{#2Chapter~#1#3}
\crefformat{part}{#2Part~#1#3}
\Crefformat{part}{#2Part~#1#3}
\crefformat{enumi}{#2(#1)#3}
\Crefformat{enumi}{#2(#1)#3}

\usepackage[mathlines]{lineno}

\newcommand*\patchAmsMathEnvironmentForLineno[1]{%
  \expandafter\let\csname old#1\expandafter\endcsname\csname #1\endcsname
  \expandafter\let\csname oldend#1\expandafter\endcsname\csname end#1\endcsname
  \renewenvironment{#1}%
     {\linenomath\csname old#1\endcsname}%
     {\csname oldend#1\endcsname\endlinenomath}}%
\newcommand*\patchBothAmsMathEnvironmentsForLineno[1]{%
  \patchAmsMathEnvironmentForLineno{#1}%
  \patchAmsMathEnvironmentForLineno{#1*}}%
\AtBeginDocument{%
  \patchBothAmsMathEnvironmentsForLineno{equation}%
  \patchBothAmsMathEnvironmentsForLineno{align}%
  \patchBothAmsMathEnvironmentsForLineno{flalign}%
  \patchBothAmsMathEnvironmentsForLineno{alignat}%
  \patchBothAmsMathEnvironmentsForLineno{gather}%
  \patchAmsMathEnvironmentForLineno{split}%%
  \patchBothAmsMathEnvironmentsForLineno{multline}}
\theoremnumbering{arabic}
\theoremstyle{plain}
\theoremsymbol{}
\theorembodyfont{\itshape}
\theoremheaderfont{\normalfont\bfseries}
\theoremseparator{.}

\newtheorem{theorem}{Theorem}
\crefformat{theorem}{#2Theorem~#1#3}
\Crefformat{theorem}{#2Theorem~#1#3}
\crefformat{cthmin}{#2Theorem~#1#3}
\Crefformat{cthmin}{#2Theorem~#1#3}

\crefformat{definition}{#2Definition~#1#3}
\Crefformat{definition}{#2Definition~#1#3}

\newcommand{\newtheoremwithcrefformat}[2]{%
  \newtheorem{#1}[theorem]{#2}%
  \crefformat{#1}{##2\MakeUppercase#1~##1##3}%
  \Crefformat{#1}{##2\MakeUppercase#1~##1##3}%
}
\newcommand{\newseptheoremwithcrefformat}[2]{%
  \newtheorem{#1}{#2}%
  \crefformat{#1}{##2\MakeUppercase#1~##1##3}%
  \Crefformat{#1}{##2\MakeUppercase#1~##1##3}%
}
\newcommand{\newclaimwithcrefformat}[2]{%
  \newtheorem{#1}{#2}[theorem]%
  \crefformat{#1}{##2\MakeUppercase#1~##1##3}%
  \Crefformat{#1}{##2\MakeUppercase#1~##1##3}%
}

\newtheoremwithcrefformat{lemma}{Lemma}

\newtheoremwithcrefformat{proposition}{Proposition}
\newtheoremwithcrefformat{observation}{Observation}
\newseptheoremwithcrefformat{conjecture}{Conjecture}
\newtheoremwithcrefformat{corollary}{Corollary}
\newclaimwithcrefformat{claim}{Claim}

\newtheorem{cthmin}{Theorem}%
\newenvironment{cthm}[1]
  {\renewcommand\thecthmin{#1}\cthmin}
  {\endcthmin}

\theoremstyle{definition}
\newtheorem*{question*}{Question}

\theorembodyfont{\upshape}
\newtheoremwithcrefformat{remark}{Remark}
\theoremstyle{nonumberplain}
\theoremheaderfont{\scshape}
\theorembodyfont{\normalfont}
\theoremsymbol{\ensuremath{\square}}

\theoremsymbol{\ensuremath{\lrcorner}}

\renewcommand{\P}{\textsf{P}}
\newcommand{\NP}{\textsf{NP}}

\newcommand{\cC}{\mathcal{C}}

\newcommand{\cH}{\mathcal{H}}

\newcommand{\cX}{\mathcal{X}}
\newcommand{\cY}{\mathcal{Y}}

\newcommand{\ork}[1]{\mathrm{OR}_{#1}}

\newcommand{\mrc}[1]{\mathrm{mrc}(#1)}

\newcommand{\N}{\mathbb{N}}

\newcommand{\R}{\mathbb{R}}

\renewcommand{\epsilon}{\varepsilon}
\newcommand{\Oh}{\mathcal{O}}

\renewcommand{\leq}{\leqslant}
\renewcommand{\geq}{\geqslant}

\renewcommand{\setminus}{-}

\newcommand{\lhomo}[1]{\textsc{LHom}(\ensuremath{#1})\xspace}
\newcommand{\lhomorc}[1]{\textsc{LHom}_{\mathrm{rc}}(\ensuremath{#1})\xspace}
\newcommand{\mchomo}[1]{\textsc{MCHom}(\ensuremath{#1})\xspace}
\newcommand{\whomo}[1]{\textsc{WHom}(\ensuremath{#1})\xspace}
\newcommand{\sat}[1]{\ensuremath{#1}\textsc{-Sat}\xspace}

\newcommand{\coloring}[1]{\ensuremath{#1}\textsc{-Coloring}\xspace}

\newcommand{\wei}{\mathfrak{w}}
\newcommand{\QQ}{\mathbb{Q}}
\newcommand{\diam}{\mathrm{diam}}
\newcommand{\len}{\mathrm{len}}
\newcommand{\IG}{\mathsf{IG}}

\newcommand{\bb}{\mathrm{bb}}
\newcommand{\area}{\mathrm{area}}

\declaretheorem[sibling=theorem]{lemma}

\begin{document}
\title{Computing list homomorphisms\\
in geometric intersection graphs}

%\title{APPENDIX: Computing list homomorphisms\\
%in geometric intersection graphs (full version of the paper)}

\date{}
\author{S\'andor Kisfaludi-Bak%
\thanks{
Department of Computer Science, Aalto University, Finland
\newline
E-mail: \texttt{sandor.kisfaludi-bak@aalto.fi}.
}
\and Karolina Okrasa% 
\thanks{
Faculty of Mathematics and Information Science, Warsaw University of Technology and Institute of Informatics, University of Warsaw
\newline
E-mail: \texttt{k.okrasa@mini.pw.edu.pl}.\newline
Supported by the European Research Council (ERC) under the European
Union’s Horizon 2020 research and innovation programme Grant Agreement no. 714704.}
\and Paweł Rzążewski%
\thanks{
 Faculty of Mathematics and Information Science, Warsaw University of Technology and Institute of
Informatics, University of Warsaw
\newline
 E-mail: \texttt{p.rzazewski@mini.pw.edu.pl}.\newline
Supported by Polish National Science Centre grant no. 2018/31/D/ST6/00062.
}}

\begin{titlepage}
\def\thepage{}
\thispagestyle{empty}
\maketitle

\begin{abstract}
A homomorphism from a graph $G$ to a graph $H$ is an edge-preserving mapping from $V(G)$ to $V(H)$.
Let $H$ be a fixed graph with possible loops.
In the list homomorphism problem, denoted by \textsc{LHom}($H$), the instance is a graph $G$,
whose every vertex is equipped with a subset of $V(H)$, called list.
We ask whether there exists a homomorphism from $G$ to $H$, such that every vertex from $G$ is mapped to a vertex from its list.

We study the complexity of the \textsc{LHom}($H$) problem in intersection graphs of various geometric objects.
In particular, we are interested in answering the question for what graphs $H$ and for what types of geometric objects,
the \textsc{LHom}($H$) problem can be solved in time subexponential in the number of vertices of the instance.

We fully resolve this question for string graphs, i.e., intersection graphs of continuous curves in the plane.
Quite surprisingly, it turns out that the dichotomy exactly coincides with the analogous dichotomy for graphs
excluding a fixed path as an induced subgraph [Okrasa, Rz\k{a}\.zewski, STACS 2021].

Then we turn our attention to subclasses of string graphs, defined as intersections of fat objects.
We observe that the (non)existence of subexponential-time algorithms in such classes is closely related to the size
$\mathrm{mrc}(H)$ of a maximum reflexive clique in $H$, i.e., maximum number of pairwise adjacent vertices, each of which has a loop.
We study the maximum value of $\mathrm{mrc}(H)$ that guarantees the existence of a subexponential-time 
algorithm for \textsc{LHom}($H$) in intersection graphs of (i) convex fat objects, (ii) fat similarly-sized objects, and (iii) disks.
In the first two cases we obtain optimal results, by giving matching algorithms and lower bounds.

Finally, we discuss possible extensions of our results to weighted generalizations of \textsc{LHom}($H$).

\end{abstract}
\end{titlepage}

\section{Introduction}
For a family $\mathcal{S}$ of sets, its intersection graph is the graph whose vertex set is  $\mathcal{S}$, and two sets are adjacent if and only if they have a nonempty intersection.
It is straightforward to observe that every graph is an intersection graph of some family of sets: each vertex can be represented by the set of incident edges. More efficient intersection representations were studied by Erd\H{o}s, Goodman, and P\'osa~\cite{erdos1966}.

A prominent role is played by \emph{geometric intersection graphs}, i.e., intersection graphs of some geometrically defined object (usually subsets of the plane). Some best studied families of this type are interval graphs~\cite{Lekkeikerker1962,GOLUMBIC2004171} (intersection graphs of segments on a line), disk graphs~\cite{DBLP:journals/dm/ClarkCJ90,DBLP:conf/waoa/Fishkin03} (intersection graphs of disks in the plane), segment graphs~\cite{KRATOCHVIL1994289} (intersection graphs of segments), or string graphs~\cite{DBLP:journals/jct/Kratochvil91,DBLP:journals/jct/Kratochvil91a} (intersection graphs of continuous curves).
Geometric intersection graphs are studied not only for their elegant structural properties, but also for potential applications.
Indeed, many real-life graphs have some underlying geometry~\cite{DBLP:conf/soda/KaufmannKLS06,JUNGCK20151,483546}. 
Thus the complexity of graph problems restricted to various classes of geometric intersection graphs has been an active research topic~\cite{Marx07a,Marx05,FominLS12,MarxP15,DBLP:conf/compgeom/FominLP0Z20,DBLP:journals/dcg/FominLPSZ19,DBLP:journals/jacm/BonamyBBCGKRST21,DBLP:journals/algorithmica/BonnetR19,BergBKMZ20}.

The underlying geometric structure can sometimes be exploited to obtain much faster algorithms than for general graphs.
For example, for each fixed $k$, the \coloring{k} problem is polynomial-time solvable in interval graphs, while for $k\geq 3$ the problem is \NP-hard and thus unlikely to be solvable in polynomial time in general graphs.
For disk graphs, the \coloring{k} problem remains \NP-hard for $k \geq 3$, but still it is in some sense more tractable than for general graphs.
Indeed, for every fixed $k \geq 3$, then \coloring{k} problem can be solved in \emph{subexponential time} $2^{\Oh(\sqrt{n})}$ in $n$-vertex disk graphs, while assuming the Exponential-Time Hypothesis (ETH)~\cite{ImpagliazzoPaturi,DBLP:journals/jcss/ImpagliazzoPZ01} no such algorithm can exist for general graphs. Furthermore, the running time of the above algorithm is optimal under the ETH~\cite{DBLP:conf/ciac/Kisfaludi-BakZ17}. 
Bir\'o et al.~\cite{DBLP:journals/jocg/BiroBMMR18} studied the problem for superconstant number of colors and showed that if $k = o(n)$, then \coloring{k} admits a subexponential-time algorithm in disk graphs, and proved almost tight complexity bounds conditioned on the ETH.

As a stark contrast, they showed that \coloring{6} does not admit a subexponential-time algorithm in segment graphs.
This was later improved by Bonnet and Rz\k{a}\.zewski~\cite{DBLP:journals/algorithmica/BonnetR19} who showed that already \coloring{4} cannot be solved in subexponential time in segment graphs, but \coloring{3} admits a $2^{\Oh(n^{2/3} \log n)}$-algorithm in all string graphs.
They also showed several positive and negative results concerning subexponential-time algorithms for segment and string graphs.

This line of research was continued in a more general setting by Okrasa and Rz\k{a}\.zewski~\cite{DBLP:journals/jcss/OkrasaR20} who considered variants of the graph homomorphism problem in string graphs. For graphs $G$ and $H$, a homomorphism from $G$ to $H$ is an edge-preserving mapping from $V(G)$ to $V(H)$. Note that a homomorphism to $K_k$ is precisely a proper $k$-coloring, so graph homomorphisms generalize colorings.
Among other results, Okrasa and Rz\k{a}\.zewski~\cite{DBLP:journals/jcss/OkrasaR20} fully classified the graphs $H$ for which a weighted variant of the homomorphism problem admits a subexponential-time algorithm in string graphs (assuming the ETH). It turns out that the substructure of $H$ that makes the problem hard is an induced 4-cycle.

\paragraph{Separators in geometric intersection graphs.}
Almost all subexponential-time algorithms for geometric intersection graphs rely on the existence of balanced separators that are small or simple in some other way. This is very convenient for a divide-\&-conquer approach -- due to the simplicity of the separator we can guess 
how the solution looks on the separator, and then recurse into connected components of the graph with the separator removed.

For example it is known that $n$-vertex disk graphs, where each point is contained in at most $k$ disks, admit a balanced separator of size $\Oh(\sqrt{nk})$~\cite{MillerTTV97,SmithW98}. Note that this result implies the celebrated planar separator theorem by Lipton and Tarjan~\cite{LiptonT79}, since by the famous Circle Packing Theroem of Koebe~\cite{koebe1936kontaktprobleme}, every planar graph is the intersection graph of internally disjoint disks (which translates to $k \leq 2$).

This separator theorem was recently significantly extended by De Berg et al.~\cite{BergBKMZ20} who introduced the notion of \emph{clique-based separators}. Roughly speaking, a clique-based separator consists of cliques, instead of measuring its size (i.e., the number of vertices), we measure its \emph{weight} defined as the sum of logarithms of sizes of the cliques, see \cref{sec:cliquebased}.
This approach shifts the focus from ``small'' separators to separators with ``simple'' structure, and proved helpful in obtaining ETH-tight algorithms for various combinatorial problems in intersection graphs of similarly sized fat or convex fat objects.
The direction was followed by De Berg et al.~\cite{BergKMT21} who proved that some other classes of intersection graphs admit balanced clique-based separators of small weight.

For general string graphs we also know a separator theorem: Lee~\cite{DBLP:conf/innovations/Lee17} proved that every string graph with $m$ edges admits a balanced separator of size $\Oh(\sqrt{m})$ (see also Matou\v{s}ek~\cite{Matousek14}), and this bound is known to be optimal. Note that each planar graph is a string graph
and has linear number of edges, so this result also implies the planar separator theorem.

In all above approaches the size (or the weight) of the separator, as well as the balance factor, were measured in purely combinatorial terms. However, some alternative approaches, with more geometric flavor, were also used.
For example Alber and Fiala~\cite{AlberF04} showed a separator theorem for intersection graphs of disks with diameter bounded from below and from above, where both the size of the separator and the size of each component of the remaining part of the graph is measured in terms of the \emph{area} occupied by the geometric representation.

\paragraph{Our contribution.} In this paper we study the complexity of the list variant of the graph homomorphism problem in intersection graphs of geometric objects. For a fixed graph $H$ (with possible loops), by \lhomo{H} we denote the computational problem,
where every vertex of the input graph $G$ is equipped with the subset of $V(H)$ called \emph{list},
and we need to determine whether there exists a homomorphism from $G$ to $H$, such that every vertex from $G$ is mapped to a vertex from its list.

First, in \cref{sec:string}, we study the complexity of \lhomo{H} in string graphs and exhibit the full complexity dichotomy,
i.e., we fully characterize graphs $H$ for which the \lhomo{H} problem can be solved in subexponential time.
It turns out that the positive cases are precisely the graphs $H$ that not \emph{predacious}. 
The class of predacious graphs was defined by Okrasa and Rz\k{a}\.zewski~\cite{DBLP:journals/jcss/OkrasaR20} who studied the complexity of \lhomo{H} in $P_t$-free graphs (i.e., graphs excluding a $t$-vertex path as an induced subgraph).
It is quite surprising that the complexity dichotomies for \lhomo{H} in $P_t$-free graphs and in string graphs coincide;
note that the classes are incomparable.
%Indeed, every path is a string graph, and the 1-subdivision of $K_5$ is $P_{10}$-free, but is not a string graph.

Our approach closely follows the one by Okrasa and Rz\k{a}\.zewski~\cite{DBLP:journals/jcss/OkrasaR20}.
First we show that if $H$ does not belong to the class of predacious graphs, then a combination of branching on a high-degree vertex and divide-\&-conquer approach using the string separator theorem yields a subexponential-time algorithm.
For the hardness counterpart, we observe that the graphs constructed in~\cite{DBLP:journals/jcss/OkrasaR20} are actually string graphs.
Summing up, we obtain the following result.

\begin{theorem}\label{thm:strings}
Let $H$ be a fixed graph.
\begin{enumerate}[(a)]
\item If $H$ is not predacious, then $\lhomo{H}$ can be solved in time $2^{\Oh( n^{2/3} \log n)}$ in $n$-vertex string graphs, even if a geometric representation is not given.
\item Otherwise, assuming the ETH, there is no algorithm for $\lhomo{H}$ working in time $2^{o(n)}$ in string graphs, even if they are given along with a geometric representation.
\end{enumerate}
\end{theorem}

Then in \cref{sec:algo} we turn our attention to subclasses of string graphs defined by intersections of fat objects.
We observe that in this case the parameter of the graph $H$ that seems to have an influence on the (non)existence of subexponential-time algorithms is the size of the maximum reflexive clique, denoted by $\mrc{H}$. Here, by a reflexive clique we mean the set of pairwise adjacent vertices, each of which has a loop.
We focus on the following question. 
\begin{adjustwidth}{1cm}{1cm}
\begin{question*}
For a class $\cC$ of geometric objects, what is the maximum $k$ (if any), such that for every graph $H$ with $\mrc{H} \leq k$,
the \lhomo{H} problem admits a subexponential-time algorithm in intersection graphs of objects from $\cC$?
\end{question*}
\end{adjustwidth}

Note that $k$ from the question might not exist, as for example \coloring{4} (and thus \lhomo{K_4}) does not admit a subexponential-time algorithm in segment graphs~\cite{DBLP:journals/algorithmica/BonnetR19}, while $\mrc{K_4}=0$.

First, we show that the existence of clique-based separators of sublinear weight is sufficient 
to provide subexponential-time algorithms for the case $\mrc{H} \leq 1$.
In particular, this gives the following result.

\begin{restatable}{theorem}{thmcliquebased}
%\begin{theorem}
\label{thm:cliquebased}
Let $H$ be a graph with $\mrc{H} \leq 1$. Then $\lhomo{H}$ can be solved in time:
\begin{enumerate}[(a)]
\item $2^{\Oh(\sqrt{n})}$ in $n$-vertex intersection graphs of fat convex objects,
\item $2^{\Oh(n^{2/3}\log n)}$ in $n$-vertex pseudodisk intersection graphs.
\end{enumerate}
provided that the instance graph is given along with a geometric representation.
%\end{theorem}
\end{restatable}

Next, we study the intersection graphs of fat, similarly-sized objects.
The exact definition of these families is given in \cref{sec:prelim},
but, intuitively, each object should contain a disk of constant diameter, 
and be contained in a disk of constant diameter.

We show that for such graphs subexponential-time algorithms exist even for the case $\mrc{H} \leq 2$.
Our proof is based on a new geometric separator theorem, which measures the size of the separator in terms of the number of vertices,
and the size of the components of the remaining graph in terms of the area.

\begin{restatable}{theorem}{thmfatssized}
%\begin{theorem}
\label{thm:fatssized}
Let $H$ be a graph with $\mrc{H} \leq 2$. Then $\lhomo{H}$ can be solved in time
$2^{\Oh( n ^{2/3} \log n)}$ in $n$-vertex intersection graphs of fat, similarly-sized objects,
provided that the instance graph is given along with a geometric representation.
%\end{theorem}
\end{restatable}

In \cref{sec:lower} we complement these results by showing that both \cref{thm:cliquebased}~(a) and \cref{thm:fatssized} are optimal in terms of the value of $\mrc{H}$. More precisely, 
we prove that there are graphs $H_1$ and $H_2$ with $\mrc{H_1}=2$ and $\mrc{H_2}=3$,
such that \lhomo{H_1} does not admit a subexponential-time algorithm in intersection graphs of equilateral triangles,
and \lhomo{H_2} does not admit a subexponential-time algorithm in intersection graphs of fat similarly-sized triangles.

A very natural question is to find the minimum value of $\mrc{H}$ that guarantees the existence of subexponential-time algorithms for \lhomo{H} in \emph{disk graphs}. By \cref{thm:cliquebased}~(a) we know that it is at least 1. However, disk graphs admit many nice structural properties that proved very useful in the construction of algorithms.
Unfortunately, we were not able to obtain any stronger algorithmic results for disk graphs.

For the lower bounds, we note that the constructions in our hardness reductions essentially used that triangles can ``pierce each other'', which cannot be done with disks (actually, even with pseudodisks, see \cref{sec:prelim} for the definition).
The best lower bound we could provide for disk graphs is the following theorem.

\begin{restatable}{theorem}{thmdiskslower}
%\begin{theorem}
\label{thm:disks-lower}
Assume the ETH.
There is a graph $H$ with $\mrc{H}=4$, such that $\lhomo{H}$ cannot be solved in time $2^{o(n/\log n)}$ in $n$-vertex disk intersection graphs, even if they are given along with a geometric representation.
%\end{theorem}
\end{restatable}
Let us point out that the construction in \cref{thm:disks-lower} is much more technically involved than the previous ones.

Finally, in \cref{sec:weighted} we study two weighted generalizations of the list homomorphism problem, called \emph{min cost homomorphism}~\cite{DBLP:journals/ejc/GutinHRY08} and \emph{weighted homomorphism}~\cite{DBLP:journals/jcss/OkrasaR20}, and show that the clique-based separator approach from \cref{thm:cliquebased} works for both of them.

The paper is concluded with several open questions in \cref{sec:conclusion}.

\section{Preliminaries} \label{sec:prelim}
All logarithms in the paper are of base 2. For a positive  integer $k$, by $[k]$ we denote $\{1,2,\ldots,k\}$. 

\paragraph{Graph theory.}
For a graph $G$ and a vertex $v \in V(G)$, by $N_G(v)$ we denote the set of neighbors of $v$.
If the graph is clear from the context, we simply write $N(v)$.
We say that two sets $A,B$ of vertices of $G$ are \emph{complete} to each other, if for every $a \in A,b \in B$ vertices $a$ and $b$ are adjacent.
For $\delta \in (0,1)$, a set a $S \subseteq V(G)$ is a \emph{$\delta$-balanced separator} if every component of $G - S$ has at most $\delta \cdot |V(G)|$ vertices.

Let $H$ be a graph with possible loops. Let $R(H)$ be the set of \emph{reflexive} vertices, i.e., the vertices with a loop, 
and let $I(H)$ be the set of \emph{irreflexive vertices}, i.e., vertices without loops. Clearly $R(H)$ and $I(H)$ form a partition of $V(H)$.
By $\mrc{H}$ we denote the size of a maximum reflexive clique in $H$.
We call $H$ a \emph{strong split graph}, if $R(H)$ is a reflexive clique and $I(H)$ is an independent set.

\paragraph{String graphs and their subclasses.}
For a set $V$ of subsets of the plane $\R^2$, by $\IG(V)$ we denote their \emph{intersection graph},
i.e., the graph with vertex set $V$ where two elements are adjacent if and only if they have nonempty intersection.
To avoid confusion, the elements of $V$ will be called \emph{objects}. 

\emph{String graphs} are intersection graphs of sets of continuous curves in the plane.
\emph{Grid graphs} are intersection graphs of axis-parallel segments such that segments in one direction are pairwise disjoint.
Note that grid graphs are always bipartite.

Two objects $a,b \subseteq \R^2$ with Jordan curve boundaries are in a \emph{pseudodisk relation} if their boundaries intersect at most twice. 
%If $a$ and $b$ are convex, then this is equivalent to $a \setminus b$ and $b \setminus a$ being arc-connected. \skb{I think they are equivalent, but I couldn't find a citation and we certainly don't want to prove it here
A collection $V$ of objects in $\R^2$ with Jordan curve boundaries is a family of \emph{pseudodisks} if all their elements are pairwise in a pseudodisk relation.
\emph{Pseudodisk} intersection graphs are intersection graphs of families of pseudodisks.

A collection $V$ of objects in $\R^2$ is \emph{fat} if there exists a constant $\alpha>0$, such that each $v \in V$ satisfies $r_{\mathrm{in},v}/r_{\mathrm{out},v}\geq \alpha$, where $r_{\mathrm{in},v}$ and $r_{\mathrm{out},v}$ denote, respectively, the radius of the largest inscribed and smallest circumscribed disk of $v$.
A collection $V$ of objects in $\R^2$ is \emph{similarly-sized} if there is some constant $\beta >0$, such that $\max_{v\in V} \diam(v) / \min_{v\in V} \diam(v)\leq \beta$, where $\diam(v)$ denotes the diameter of $v$.
If a collection of objects is fat and similarly-sized, then we can set the unit to be the smallest diameter among the maximum inscribed disks of the objects, and as a consequence of the properties each object can be covered by some disk of radius $R=\Oh(1)$.

\paragraph{List homomorphisms.}
A \emph{homomorphism} from a graph $G$ to a graph $H$ is a mapping $f : V(G) \to V(H)$, such that for every $uv \in E(G)$ it holds that $f(u)f(v) \in E(H)$. If $f$ is a homomorphism from $G$ to $H$, we denote it shortly by $f : G \to H$.
Note that homomorphisms to the complete graph on $k$ vertices are precisely proper $k$-colorings.
Thus we will often refer to vertices of $H$ as \emph{colors}.

In this paper we consider the \lhomo{H} problem, which asks for the existence of \emph{list} homomorphisms.
Formally, for a fixed graph $H$ (with possible loops), an instance of $\lhomo{H}$ is a pair $(G,L)$, where $G$ is a graph and $L : V(G) \to 2^{V(H)}$ is a list function.
We ask whether there exists a homomorphism $f: G \to H$, which respects the lists $L$, i.e., for every $v \in V(G)$ it holds that $f(v) \in L(v)$. If $f$ is such a list homomorphism, we denote it shortly by $f : (G,L) \to H$.
We also write $(G,L) \to H$ to denote that some such $f$ exists.

Two vertices $a,b$ of $H$ are called \emph{comparable} if $N(a) \subseteq N(b)$ or $N(b) \subseteq N(a)$; otherwise $a$ and $b$ are \emph{incomparable}. 
A subset $A$ of vertices of $H$ is \emph{incomparable} if it contains pairwise incomparable vertices.

Let $a,b \in V(H)$ be such that $N(a) \subseteq N(b)$. 
Note that in any homomorphism $f : G \to H$ and any $v \in V(G)$ such that $f(v)=a$,
recoloring $v$ to the color $b$ yields another homomorphism from $G$ to $H$.
Thus for any instance $(G,L)$  of \lhomo{H}, if some list $L(v)$ contains two vertices $a,b$ as above, we can safely remove $a$ from the list. 
Consequently, we can always assume that each list is incomparable.

The following straightforward observation will be used several times in the paper.

\begin{observation}\label{lem:mapclique}
Let $G$ be an irreflexive graph and let $H$ be a graph with possible loops and let $f : G \to H$.
For every clique $C$ of $G$ we have the following:
\begin{itemize}
\item at most $|I(H)|$ vertices from $C$ are mapped to vertices of $I(H)$ (each to a distinct vertex of $I(H)$),
\item the remaining vertices of $C$ are mapped to some reflexive clique of $H$.
\end{itemize}
\end{observation}

\newpage
\section{Dichotomy for string  graphs}\label{sec:string}
The complexity dichotomy for the list homomorphism problem (in general graphs) was shown by Feder, Hell, and Huang~\cite{DBLP:journals/jgt/FederHH03} (see also~\cite{FEDER1998236,DBLP:journals/combinatorica/FederHH99}),
who showed that if $H$ is a so-called \emph{bi-arc graph}, then \lhomo{H} is polynomial-time solvable, and otherwise it is \NP-complete. Furthermore their hardness reductions imply that the problem cannot be solved in subexponential time, assuming the ETH.

One of the ways to define bi-arc graphs is via their \emph{associate bipartite graphs}.
For a graph $H$ with possible loops, is bipartite associated graph is the graph $H^*$ constructed as follows.
For each vertex $a \in V(H)$ we introduce two vertices, $a',a''$ to $V(H^*)$, and $a'b'' \in E(H^*)$ if $ab \in E(H)$.
Now $H$ is a bi-arc graph if and only if $H^*$ is bipartite and its complement is a circular-arc graph, i.e., it is an intersection graph of arcs of a circle.

The key idea of the \NP-hardness proof by Feder at al.~\cite{DBLP:journals/combinatorica/FederHH99} was to reduce showing hardness for \lhomo{H}, where $H$ is not a bi-arc graph, to showing hardness to \lhomo{H^*}, where $H^*$ is not the complement of a co-bipartite circular-arc graph.
We will use a similar approach, so let us explain it in more detail. We use the terminology of Okrasa {et al.}~{\cite{FullerComplexity}.

Let $(G, L^*)$ be an instance of \lhomo{H^*} problem, and let $X,Y$ be bipartition classes of $H^*$. 
We say that $(G, L^*)$ is \emph{consistent} if $G$ is bipartite with bipartition classes $G_X, G_Y$, and either for every $x \in G_X$ (resp. $y \in G_Y$) we have $L(x) \subseteq X$ (resp. $L(y) \subseteq Y$), or  for every $x \in G_X$ (resp. $y \in G_Y$) we have $L(x) \subseteq Y$ (resp. $L(y) \subseteq X$).

\begin{lemma}[Okrasa {et al.}~{\cite{FullerComplexity}}]\label{lem:homo-star}
Let $H$ be a graph and let $(G,L^*)$ be a consistent instance of \lhomo{H^*}
Define $L: V(G) \to 2^{V(H)}$ as follows: for every $v \in V(G)$ we have $L(v):= \{a ~|~ \{a',a''\} \cap L^*(v) \neq \emptyset\}$.
Then $(G,L^*) \to H^*$ if and only if $(G,L) \to H$.
\end{lemma}

Let us recall one more result of Okrasa {et al.}~{\cite{FullComplexity,FullerComplexity}.
We note that it uses a certain notion of ``undecomposability'' of a bipartite graph. 
Since the property of undecomposability will only be used to invoke another known result (\cref{thm:gadgets}),
we omit the technically involved definition (an inquisitive reader will find it in \cite{FullComplexity}).

\begin{theorem}[Okrasa et al.~\cite{FullComplexity,FullerComplexity}]\label{thm:factorization}
Let $H$ be a graph.
In time $|V(H)|^{\Oh(1)}$ we can construct a family $\cH$ of $\Oh(|V(H)|)$ connected graphs, called \emph{factors of $H$}, such that:
\begin{enumerate}[(1)]
\item $H$ is a bi-arc graph if and only if every $H' \in \cH$ is a bi-arc graph, \label{factors:hard}
%\item if $H$ is bipartite, then each $H' \in \cH$ is an induced subgraph of $H$, and is either the complement of a circular-arc graphs or is undecomposable, \label{factors:bipartite} \ko{czy wywalic}
%\item otherwise, 
\item for each $H' \in \cH$, the graph $H'^*$ is an induced subgraph of $H^*$ and at least one of the following holds: \label{factors:types}
\begin{enumerate}
\item $H'$ is a bi-arc graph, or
\item $H'$ a strong split graph, and has an induced subgraph $H''$, which is not a bi-arc graph and is an induced subgraph of $H$, or \label{factors:types:strongsplit}
\item $(H')^*$ is undecomposable, \label{factors:types:undecomposable}
\end{enumerate}
\item for every instance $(G,L)$ of \lhomo{H}, the following implication holds: \label{factors:bottomup}

If there exists a non-decreasing, convex function $f \colon \N \to \R$,
such that for every $H' \in \cH$, for every induced subgraph $G'$ of $G$, and for every $L': V(G') \to 2^{|V(H')|}$,
we can decide whether $(G',L')\to H'$ in time $f(|V(G')|)$, then
we can solve the instance $(G,L)$  in time
\[
\Oh \left (|V(H)| f(n) + n^2 \cdot |V(H)|^3 \right).
\]
\end{enumerate}
\end{theorem}

Let $H$ be a graph and let $\cH$ be as in \cref{thm:factorization}.
We say that $H$ is \emph{predacious}, if there exists a factor $H' \in \cH$ that is not bi-arc, and two incomparable two-element sets $\{a_1,a_2\},\{b_1,b_2\} \subseteq V(H')$, such that $\{a_1,a_2\}$ is complete to $\{b_1,b_2\}$.
We say that the tuple $(a_1,a_2,b_1,b_2)$ is a \emph{predator}.

\subsection{Algorithm}

We start by proving part (a) of~\cref{thm:strings}.
We present an algorithm that combines the approach of exploiting the absence of predators in the target graph that was used in~\cite{DBLP:conf/stacs/OkrasaR21} with the following string separator theorem.

\begin{theorem}[Lee~\cite{DBLP:conf/innovations/Lee17}, Matou\v{s}ek~\cite{Matousek14}]\label{thm:string-separator}
Every string graph with $m$ edges has a $\frac{2}{3}$-balanced separator of size $\Oh(\sqrt{m})$. It can be found in polynomial time, if the geometric representation is given. 
\end{theorem}

Now we are ready to show the algorithmic part of \cref{thm:strings}.

\begin{cthm}{1 (a)}%\begin{theorem}
\label{thm:string-non-predacious}
Let $H$ be a connected graph which does not contain a predator. Then the \lhomo{H} problem can be solved in time $2^{O(n^{2/3}\log n)}$.
%\end{theorem}
\end{cthm}
\begin{proof}
Let $(G,L)$ be in instance of \lhomo{H} with $n$ vertices.
We present a recursive algorithm; clearly we can assume that the statement holds for constant-size instances.

Each step of our algorithm starts with preprocessing the instance by repeating steps 1.-3. exhaustively, in the given order.
\begin{enumerate}
\item We make every list an incomparable set, i.e., if there exists $v \in V(G)$ and $a,b \in L(v)$ such that $N_H(a) \subseteq N_H(b)$, we remove $a$ from $L(v)$.
\item For every pair of vertices $u,v \in V(G)$, if there exist $a \in L(u)$ such that for every $b \in L(v)$ it holds that $ab \notin E(H)$, we remove $a$ from $L(u)$. 
The correctness of this step follows from the fact that there is no $h: (G,L) \to H$ such that $h(u)=a$, as otherwise $h(u)h(v) \notin E(H)$.
\item If there exists $v \in V(G)$ such that $|L(v)|=1$, we remove $v$ from $G$.
This step is correct, because in step 2. we already adjusted the lists of the neighbors of $v$ to contain only neighbors of $a$.  
\end{enumerate} 
If at any moment during the preprocessing phase a list of any vertex becomes empty, we terminate the current branch and report that $(G,L)$ is a no-instance. Clearly the rules above can be applied in polynomial time.

Suppose now that none of the preprocessing rules can be applied.
For simplicity let us keep calling the obtained instance $(G,L)$ and denoting its number of vertices by $n$.
We distinguish two cases.

Assume that there exists a vertex $v \in V(G)$ such that $\deg(v)\geq n^{1/3}$.
Observe that there are at most $2^{|V(H)|}$ distinct lists assigned to the vertices of $G$ and thus there exists a list $L'$ that is assigned to at least $\ell := n^{1/3}/2^{|V(H)|}$ neighbors of $v$.
Since the third preprocessing rule cannot be applied, we know that each of $L(v)$ and $L'$ has at least two elements, and these elements are incomparable, as the first rule cannot be applied.

We observe that there exist $a \in L(v)$ and $b \in L'$ such that $ab \notin E(H)$.
Indeed, otherwise the tuple $(a_1,a_2,b_1,b_2)$ such that $a_1,a_2 \in L(v)$ and $b_1,b_2 \in L'$ is a predator in $H$, a contradiction. 
We branch on assigning $a$ to $v$: we call the algorithm twice, in one branch removing $a$ from $L(v)$, and in the other setting $L(v) = \{a\}$.
Note that in the preprocessing phase in the second call $b$ gets removed from at least $\ell$ lists of neighbors of $v$.

Denoting $N:=\sum_{v \in V(G)}|L(v)|$, the complexity of this step is given by the recursive inequality
\[F(N) \leq F(N-1) + F(N-\ell) \leq \ell^{O(N/\ell)} = 2^{O((N\log\ell)  /\ell)}.\]  
As $N \leq |V(H)| \cdot n$, we obtain that the running time in this case is $2^{\Oh(n^{2/3}\log n)}$.

In the second case, if no vertex of $G$ has degree larger than $n^{1/3}$, then the number of edges in $G$ is bounded by $n^{4/3}$.
By \cref{thm:string-separator} there exists a balanced separator of size $\Oh(\sqrt{n^{4/3}})=\Oh(n^{2/3})$ in $G$.
We find $S$ in time $n^{O(n^{2/3})}$ by exhaustive guessing.
Then we consider all possible list $H$-colorings of $S$, for each we update the list of neighbors of colored vertices and run the algorithm recursively for every connected component of $G-S$.
The complexity of this step is also $2^{\Oh(n^{2/3}\log n)}$, and so is the overall complexity of the algorithm.
\end{proof}

Let us point out that the running time can be slightly improved if the graph is given along with its geometric representation.
as then in the second case the balanced separator can be found in polynomial time.
Setting the degree threshold to $n^{1/2} \log^{2/3}n$ yields the running time $2^{\Oh(n^{2/3} \log^{1/3} n)}$.

\subsection{Lower bound}

In this section we complete~\cref{thm:strings} by showing the lower bound for predacious target graphs $H$.
Let $H$ be a predacious graph, let $\cH$ be as in \cref{thm:factorization}, and let $H' \in \cH$ be a factor of $H$ that is non-bi-arc, and contains a predator.
Observe that by \cref{thm:factorization} two cases may happen: either (i) $H'$ is a strong split graph, and has an induced subgraph $H''$,
which is not a bi-arc graph and is an induced subgraph of $H$,
or (ii) $H'^*$ is non-bi-arc and undecomposable. We show that in both of these cases \cref{thm:strings}~(b) holds.

First, assume that (i) is satisfied. 
Observe that since $H''$ is an induced subgraph of $H'$, $H''$ is a strong split graph. 
Therefore, \cref{thm:strings}~(b) is implied by the following theorem.

\begin{theorem}[\cite{DBLP:conf/stacs/OkrasaR21}]\label{thm:splits}
Let $H''$ be a fixed non-bi-arc strong split graph.
Then the \lhomo{H''} problem cannot be solved in time $2^{o(n)}$ in $n$-vertex split graphs, unless the ETH fails.
\end{theorem}

Indeed, since split graphs form a subclass of string graphs (see \cref{fig:string-split}), \cref{thm:splits} implies that \lhomo{H''} cannot be solved in time $2^{o(n)}$ in $n$-vertex string graphs.
Since $H''$ is an induced subgraph of $H$, each instance of \lhomo{H''} is also an instance of \lhomo{H}.
Hence, if we could solve \lhomo{H} in time $2^{o(n)}$, this would contradict the ETH by~\cref{thm:splits}.

\begin{figure}[htb]
\centering
\includegraphics[scale=1,page=4]{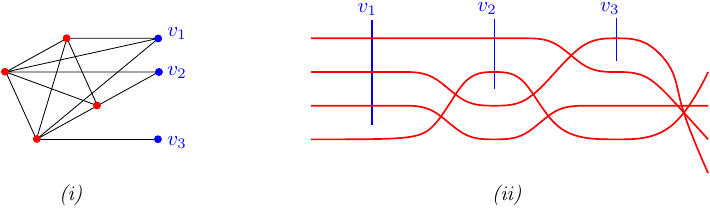}
\caption{(i) An example of a split graph $G$ and (ii) a string representation of $G$. It is straightforward to see that an analogous construction as in (ii) can be used to represent any split graph.}
\label{fig:string-split}
\end{figure}

Therefore, we can assume that (ii) holds, i.e., $H'^*$ is undecomposable. 
Observe that with \cref{lem:homo-star} at hand it is enough to prove the following.

\begin{theorem}\label{thm:homo-string-bipartite}
Let $H'$ be a graph that contains a predator, and such that $H'^*$ is non-bi-arc and undecomposable.
Assuming the ETH, the \lhomo{H'^*} problem cannot be solved in time $2^{o(n)}$ in $n$-vertex grid graphs that are consistent instances of $\lhomo{H'^*}$, even if they are given along with a geometric representation.
\end{theorem}

Let us show that \cref{thm:homo-string-bipartite} implies \cref{thm:strings}.
\medskip

\noindent\textbf{(\cref{thm:homo-string-bipartite} $\to$ \cref{thm:strings})} 
Assume the ETH, and suppose that \cref{thm:homo-string-bipartite} holds and \cref{thm:strings} (b) does not, i.e., there is an algorithm $A$ that solves \lhomo{H} in $n$-vertex string graphs in time $2^{o(n)}$.
Let $(G,L^*)$ be a consistent instance of \lhomo{H'^*}.
Since $H'^*$ is an induced subgraph of $H^*$, the instance $(G,L^*)$ is also an instance of \lhomo{H^*}.
We create an instance $(G,L)$ of \lhomo{H} as in \cref{lem:homo-star}; note that $G$ is a string graph, since grid graphs form a subclass of string graphs.
Since $(G,L^*) \to H^*$ if and only if $(G,L) \to H$, we can use $A$ to decide whether $(G,L^*) \to H'^*$ in time $2^{o(n)}$, a contradiction.

\medskip

Therefore, it remains to prove \cref{thm:homo-string-bipartite}.
Before we proceed to the reduction, we introduce two types of gadgets that will be used in our construction.

Let $a,b$ be distinct vertices of $V(H)$.
An \emph{$\ork{3}(a,b)$-gadget} is a triple $(F,L,\{o_1,o_2,o_3\})$, such that $(F,L)$ is an instance of \lhomo{H} and $o_1,o_2,o_3$ are vertices of $F$ such that $L(o_1)=L(o_2)=L(o_3)=\{a,b\}$,~and
\[
\{f(o_1)f(o_2)f(o_3) ~|~  f: (F,L) \to H \}  =\{aaa,aab,aba,baa,abb,bab,bba\}.
\]
In other words, if we consider all mappings $\{o_1,o_2,o_3\} \to \{a,b\}$, the mapping that assigns $b$ to each vertex is the only one that cannot be extended to a list homomorphism of $(F,L)$.

Let $a,b,c,d \in V(H)$.
An \emph{$(a/b \to c/d)$-gadget} is a tuple $(F,L,\{p,q\})$ such that $(F,L)$ is an instance of \lhomo{H} and $p,q \in V(F)$ such that $L(p)=\{a,b\}$, $L(q)=\{c,d\}$,~and
\[
\{f(p)f(q) ~|~  f: (F,L) \to H \}  =\{ac,bd\}.
\]

In both cases we call $F$ the \emph{underlying graph} of the gadget and vertices $o_1,o_2,o_3,p,q$ the \emph{interface vertices}.

For $t_1,t_2,t_3 > 0$ by $S_{t_1,t_2,t_3}$ we denote the graph that consists of three induced paths $P_{t_1}$, $P_{t_2}$, and $P_{t_3}$, and an  additional vertex $c$ that is adjacent to one endvertex of each path. We call $c$ the \emph{central vertex}.

We use~\cite{DBLP:conf/stacs/OkrasaR21} to obtain the gadgets which will be needed.

\begin{lemma}[\cite{DBLP:conf/stacs/OkrasaR21}]
\label{thm:gadgets} 
Let $H$ be a connected, bipartite, non-bi-arc, undecomposable graph.
Let $\{a,b\},\{c,d\} \subseteq V(H)$ be incomparable sets, each contained in one bipartition class. 
Then the following conditions hold.
\begin{enumerate}[(1)]
\item \label{lab:d-gadget} There exist an $(a/b \to c/d)$-gadget whose underlying graph is an even cycle.
\item \label{or-gadget} There exist an $\ork{3}(a,b)$-gadget whose underlying graph is $S_{t_1,t_2,t_3}$ such that $t_1,t_2,t_3>0$ are even and the interface vertices are of degree 1.
%\item \label{neq-gadget} There exists an $\neq(a,b)$-gadget whose underlying graph is a cycle.
\end{enumerate}
Furthermore, both gadgets are consistent instances of $\lhomo{H}$.
\end{lemma}

We proceed to the proof of \cref{thm:homo-string-bipartite}.
\begin{proof}[Proof of \cref{thm:homo-string-bipartite}]
First, we observe that if there exists a predator $(a_1, a_2, b_1, b_2)$ in $H'$, then $(a_1',a_2',b_1'',b_2'')$ must be a predator in $H'^*$ (we use $'$ and $''$ as in the definition of associate bipartite graph).
Let $X$ and $Y$ be the bipartition classes of $H'^*$, so that $a'_1$, $a'_2 \in X$, and $b''_1, b''_2 \in Y$.

We reduce from \sat{3}. 
Let $\Phi=(X,C)$ be an instance of \sat{3} with variables $X=\{x_1,\ldots,x_N\}$ and clauses $C=\{C_1,\ldots,C_M\}$.
The ETH with the Sparsification Lemma imply that there is no algorithm solving every such an instance in time $2^{o(N+M)}$~\cite{ImpagliazzoPaturi,DBLP:journals/jcss/ImpagliazzoPZ01}.
We can assume that each clause contains exactly three literals, since if some $C_i$ is shorter, we can duplicate some literal that already belongs to $C_i$. In what follows we assume that the literals within each clause are ordered.
 
We need to construct an instance $(G_\Phi,L)$ of \lhomo{H'^*} such that $(G_\Phi,L) \to H'^*$ if and only if there exists an assignment $\gamma: X \to \{0,1\}$ that satisfies all the clauses. 
We describe the construction of $G_\Phi$ by its segment representation and assign the lists $L$ to the segments.

First, we introduce a set $V=\{v_1,\ldots,v_N\}$ of pairwise disjoint horizontal segments and a set $U=\{u_1,\ldots,u_{3M}\}$ of pairwise disjoint vertical segments, placed so that segments from $V \cup U$ form a grid (see \cref{fig:string-general} (left)). Note that the set $U \cup V$ corresponds to a biclique in $G_\Phi$. 

Intuitively, each segment $v_i$ corresponds to the variable $x_i$ and each segment from $U$ corresponds to one literal occurring in a clause.
We define the partition of $U$ into three-element subsets $U_1,U_2,\ldots,U_{3M}$ as follows: $U_1$ consists of the three leftmost segments of $U$, $U_2$ consists of the next three segments and so on. For each $i \in [M]$, the segments from $U_i$ correspond to the (ordered) literals occurring in the clause $C_i$.

\begin{figure}[htb]
\centering
\includegraphics[scale=1,page=3]{figs-ko}

\includegraphics[scale=1,page=2]{figs-ko}\hskip 1cm
\includegraphics[scale=1,page=1]{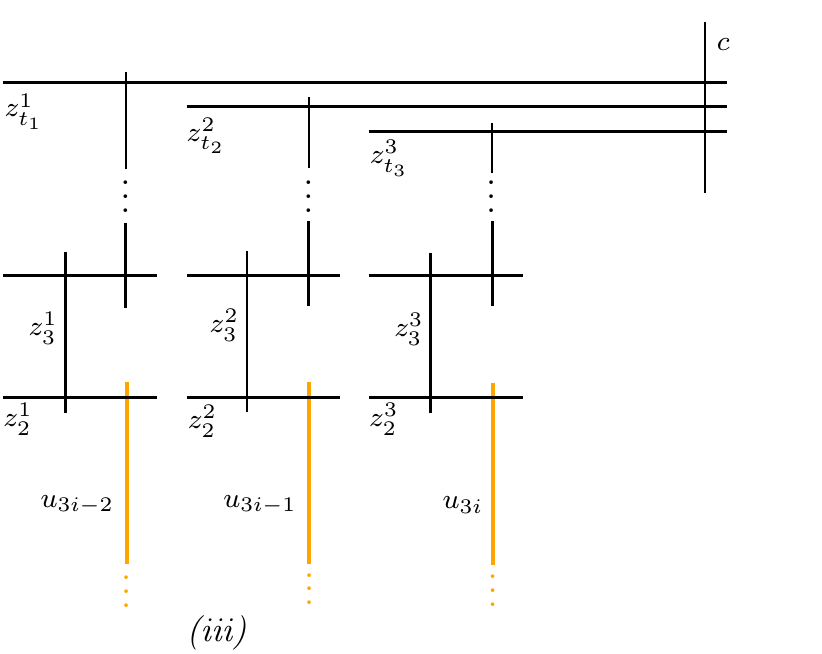}
\caption{(i) A general scheme of the construction of geometric representation of $G_\Phi$. 
Green and red boxes correspond to the places where we insert occurrence gadgets (positive or negative), yellow boxes correspond to the places where we insert $\ork{3}(b''_1,b''_2)$-gadgets.
(ii) a segment representation of the occurrence gadget. (iii) a segment representation of an $\ork{3}(b''_1,b''_2)$-gadget.}
\label{fig:string-general}
\end{figure}

For every $v_j \in V$ we set $L(v_j)=\{a'_1,a'_2\}$, and for every $u_i \in U$ we set $L(u_i)=\{b''_1,b''_2\}$.
Mapping $v_j$ to $a'_1$ (resp. to $a'_2$) will correspond to setting $x_j$ to be true (resp. false).
Mapping $u_i$ to $b''_1$ (resp. to $b''_2$) will correspond to the appropriate literal being true (resp. false).

Now we need to ensure that
\begin{enumerate}[(P1)]
\item the mapping of literals agrees with the mapping of corresponding variables, and
\item each clause contains a true literal.
\end{enumerate} 

We ensure property (P1) by introducing \emph{occurrence gadgets}.
For every $v \in V$ and $u \in U$, such that $u$ corresponds to a positive (resp. negative) occurence of $v$,
we construct a positive (resp., negative) occurence gadget, which is the $(a'_1/a'_2 \to b''_1/b''_2)$-gadget (resp. the $(a'_1/a'_2 \to b''_2/b''_1)$-gadget) given by \cref{thm:gadgets} with interface vertices $p,q$.
Consider one such gadget $(F,L,\{p,q\})$.
Recall from \cref{thm:gadgets} that $F$ is a cycle; denote its consecutive vertices $p,c_1,\ldots,c_k,q,c_{k+1},\ldots,c_\ell$ (with the possibility that the part $c_{k+1},\ldots,c_\ell$ is empty, if $p$ and $q$ are adjacent).
Observe that $L(p)=\{a'_1,a'_2\}=L(v)$ and $L(q)=\{b''_1,b''_2\}=L(u)$, so we can identify $p$ with $v$ and $q$ with $u$.
Note that $\{a'_1,a'_2\}$ is complete (in $H$) to $\{b''_1,b''_2\}$, so the edge $vu$ does not impose any further restriction on list homomorphisms of the graph we are constructing.

We represent the  parts of the cycle between $p$ and $q$ as sequences of segments,
and insert them near the intersection of $v$ and $u$, as shown in \cref{fig:string-general} (middle).

Finally, let us show how to ensure property (P2).
Consider a clause $C_i$ and the set $U_i = \{u_{3i-2}, u_{3i-1}, u_{3i}\}$ corresponding to the literals of $C_i$.
We use \cref{thm:gadgets} again, to deduce that we can construct an $\ork{3}(b''_1,b''_2)$-gadget whose underlying graph is $S_{t_1,t_2,t_3}$, where $t_1,t_2,t_2$ are even.
Let $c$ be the central vertex of $S_{t_1,t_2,t_3}$, and for every $j \in [3]$ let $z^j_1,\ldots,z^j_{t_j}$ be consecutive vertices of the $j$-th path of $S_{t_1,t_2,t_3}$ (so $z^j_1=u^j$ and $z^j_{t_j}$ is adjacent to $c$). Note that since $t_1,t_2$, and $t_3$ are even, $c$ belongs to the same bipartition class as $z^1_1,z^2_1$, and $z^3_1$.
We represent the graph $S_{t_1,t_2,t_3}$ as three sequences of segments intersecting a horizontal line $c$, as shown in \cref{fig:string-general} (right), and insert them above the three segments $u_{3i-2}, u_{3i-1}, u_{3i}$.
That concludes the construction of $G_\Phi$.
Note that all vertices of $G_\Phi$ that are not in $U \cup V$ come from the gadgets, so they are already equipped with lists. 

From the above considerations it is straightforward to verify that $(G_\Phi,L) \to H'^*$ if and only if there exists a satisfying assignment for $\Phi$. Moreover, note that by the properties of gadgets given by \cref{thm:gadgets}, the constructed instance is consistent.
To conclude the proof, observe that the gadgets we introduced are of constant size, and their number is linear in $N+M$, so the lower bound follows.
\end{proof}

\newpage
\section{Algorithms for intersection graphs of fat objects}\label{sec:algo}
\subsection{Graph classes admitting clique-based separators} \label{sec:cliquebased}
For a constant $\delta <1$, a \emph{$\delta$-balanced clique-based separator} in  graph $G$ is a family $C=\{C_1,C_2,\ldots,C_p\}$, of subsets of $V(G)$, such that:
\begin{itemize}
\item $\bigcup_{i=1}^p C_i$ is a $\delta$-balanced separator in $G$,
\item for each $i \in [p]$, the set $C_i$ induces a clique of $G$.
\end{itemize}
The \emph{weight} $w(C)$ of a clique-based separator $C=\{C_1,C_2,\ldots,C_p\}$ is defined as $\sum_{i=1}^p \log (|C_i|+1)$.
For a function $f$, we say that a class $\cC$ of graphs \emph{admits balanced clique-based separators of weight $f$},
if there is some $\delta < 1$, such that every $G \in \cC$ with $|V(G)|=n$ admits a $\delta$-balanced clique-based separator of weight at most $f(n)$.

\begin{theorem} \label{thm:cliquebased-general}
Let $H$ be a graph with $\mrc{H} \leq 1$. Let $\alpha <1$ be a constant.
Let $\cC$ be a hereditary class that admits balanced clique-based separators of weight $\Oh(n^{\alpha})$, which can be computed in time $2^{\Oh(n^\alpha)}$.
Then $\lhomo{H}$ in $n$-vertex graphs in $\cC$ can be solved in time $2^{\Oh(n^\alpha)}$.
\end{theorem}
\begin{proof}
Consider an $n$-vertex graph $G \in \cC$. We can assume that $n$ is large (in particular, $n \geq |V(H)|$), as otherwise we can solve the problem by brute-force.
Let $C=\{C_1,C_2,\ldots,C_p\}$ be a balanced clique-based separator of $G$ with weight $w(C) = \Oh(n^\alpha)$.
By our assumption, we can find it in time $2^{\Oh(n^\alpha)}$.

%Consider one clique $C_i$ of $C$.
%Recall from \cref{lem:mapclique} that at most $|I(H)|$ vertices from $C_i$ are mapped to $I(H)$,
%and the remaining vertices are mapped to some reflexive clique of $H$. However, since $\mrc{H} \leq 1$, this means that the remaining vertices (if any) are mapped to a single vertex from $R(H)$.
%Thus the total number of possible colorings of $C_i$ is at most $|C_i|^{|I(H)|} \cdot |R(H)| \leq |C_i|^{|V(H)|}$.\ko{first bound does not work for mrc=0}
%
%Consequently, the total number of colorings of all cliques in $C$ is at most
%\[
%\prod_{i=1}^p  |C_i|^{|V(H)|} = \prod_{i=1}^p  2^{|V(H)| \cdot \log (|C_i|)} = 2^{|V(H)| \cdot \sum_{i=1}^p \log |C_i|} \leq 2^{|V(H)| \cdot w(C)} = 2^{\Oh(n^\alpha)}.
%\]

Consider one clique $C_i$ of $C$.
If $|C_i| \leq |I(H)|$, then there are at most $|V(H)|^{|C_i|} \leq 2^{|V(H)| \log |V(H)|} \leq 2^{|V(H)| \log |V(H)| \log (|C_i|+1)}$ ways to map the vertices from $C_i$ to the vertices of $H$. Suppose now that $|C_i| > |I(H)|$.
Recall from \cref{lem:mapclique} that at most $|I(H)|$ vertices from $C_i$ are mapped to $I(H)$,
and the remaining vertices are mapped to some reflexive clique of $H$. However, since $\mrc{H} \leq 1$, this means that the remaining vertices (if any) are mapped to a single vertex from $R(H)$.
Thus the total number of possible colorings of $C_i$ is at most $|C_i|^{|I(H)|} \cdot |R(H)| \leq |C_i|^{|V(H)|} \leq 2^{|V(H)| \log |V(H)| \log |C_i|} \leq 2^{|V(H)| \log |V(H)| \log (|C_i|+1)}$.

Consequently, the total number of colorings of all cliques in $C$ is at most
\[
\prod_{i=1}^p  2^{|V(H)| \log |V(H)| \cdot \log (|C_i|+1)} = 2^{|V(H)|
\log|V(H)| \cdot \sum_{i=1}^p \log (|C_i|+1)} \leq 2^{|V(H)|\log |V(H)| \cdot w(C)} = 2^{\Oh(n^\alpha)}.
\]

Now we proceed using a standard divide-and-conquer approach. We exhaustively guess the coloring of the separator,
update the lists of neighbors of the vertices whose colors were guessed, and solve the subproblem in each connected component of $G - \bigcup_{i=1}^p C_i$ independently. The total running time is given by the recursive inequality for some $\delta < 1$.
\[
F(n) \leq 2^{\Oh(n^\alpha)} + 2^{\Oh(n^\alpha)} F(\delta \cdot n).
\]
This is solved by $2^{\Oh(n^\alpha)}$, which completes the proof.
\end{proof}

Now \cref{thm:cliquebased} follows directly from \cref{thm:cliquebased-general}.

\thmcliquebased*
\begin{proof}
The first statement follows from the fact that intersection graphs of fat convex objects admit balanced clique-based separators of weight $\Oh(\sqrt{n})$, which can be found in polynomial time, if the geometric representation is given~\cite{BergBKMZ20}.

Similarly, for the second statement, pseudodisk intersection graphs of fat convex objects admit balanced clique-based separators of weight $\Oh(n^{2/3}\log n)$, which can be found in polynomial time, if the geometric representation is given~\cite{BergKMT21}.
\end{proof}

\subsection{Fat, similarly-sized objects} \label{sec:algo-fatssized}
In this section we consider intersection graphs of fat, similarly-sized objects.
The algorithm presented in this section uses the area occupied by the geometric representation as the measure of the instance.
Let us start with introducing some notions.

Let $V$ be a set of $n$ fat, similarly-sized objects in $\R^2$.
Recall that there is a constant $R$, such that each object in $V$ contains a unit diameter disk and is contained in a disk of radius $R$.
In what follows we hide the factors depending on $R$ in the $\Oh(\cdot)$ notation.

Let us imagine a fine grid partitioning of $\R^2$ into square cells of unit diameter, i.e., of side length $1/\sqrt{2}$.
This lets us to use discretized notion of a bounding box and of the area.

For an object $v \subseteq \R^2$, by $\bb(v)$ we denote the minimum grid rectangle (i.e., rectangle whose sides are contained in grid lines) containing $v$, and by $\area(v)$ we denote the area of $\bb(v)$.
For a set $V$ of objects, we define $\bb(V) = \bb(\bigcup_{v \in V} v)$ and $\area(V) = \area(\bb(V))$.

Let us point out that in general $\area(V)$ can be arbitrarily large (unbounded in terms of $n$). 
However, it is straightforward to observe that this is not the case if $\IG(V)$ is connected.

\begin{observation}\label{obs:connectedarea}
Let $V$ be a set of $n$ fat, similarly-sized objects in $\R^2$, such that $\IG(V)$ is connected.
Then $\area(V) = \Oh(n^2)$.
\end{observation}

Recall that each object $v \in V$ contains a unit-diameter disk with the center $c_v$.
We assign $v$ to the grid cell containing $c_v$ (if $c_v$ is on the boundary of two or more cells, we choose one arbitrarily).
Now note that all sets assigned to a single cell form a clique in $\IG(V)$.
Consequently, the vertex set of $\IG(V)$ can be partitioned into $\Oh(\area(V))$ subsets, each inducing a clique; we call these subsets \emph{cell-cliques}.

\paragraph*{Reduction to small-area instances.}
First, we show that intersection graphs of fat, similarly-sized objects admit balanced separators, where the size of instances is measured in terms of the \emph{area} occupied by the geometric representation.

\begin{lemma} \label{lem:trueseparator}
Let $G=\IG(V)$, where $V$ is a set of $n$ fat, similarly-sized objects in $\R^2$ and $G$ is connected.
Then either $\area(V) = \Oh(n^{2/3})$, or there exists a horizontal or vertical separating line $\ell$ such that:
\begin{itemize}
\item the number of objects whose convex hull intersects $\ell$ is $\Oh(n^{2/3})$, and
\item the sets $V_1,V_2$ of objects on each side of $\ell$ (whose convex hulls are disjoint from $\ell$) satisfy
\[
\area(V_1) \leq \frac{3}{4} \area(V) \text{ and }\area(V_2) \leq \frac{3}{4} \area(V).
\]
\end{itemize}
Furthermore $\ell$ can be found in time polynomial in $\area(V)$ and $n$.
\end{lemma}

\begin{proof}
Suppose that $\area(V)$ is $\Omega(n^{2/3})$, and let $a$ (resp., $b$) be the number of vertical (resp., horizontal) gird lines intersecting $\bb(V)$.

Notice that $\area(V) = \Theta((a-1)(b-1))$, thus $ab=\Omega(n^{2/3})$.
Consequently, $\max(a,b)=\Omega(n^{1/3})$; assume without loss of generality that $a=\Omega(n^{1/3})$.
Let $\ell_1,\ell_2,\ldots,\ell_a$ be the vertical grid lines intersecting $\bb(V)$, ordered from left to right.
Observe that the convex hull of an object that can be covered by a disk of radius $R$ can intersect at most $2R\cdot\sqrt{2}$ vertical grid lines. Let $t=\lceil 2R\cdot\sqrt{2} \rceil$.
Then we have that the convex hull of each object in $V$ contributes to at most $t$ distinct vertical lines.
Hence, the total number of intersections between the vertical lines $\ell_{\lceil (a-1)/3 \rceil},\dots,\ell_{ \lfloor 2(a-1)/3 \rfloor}$ and the convex hulls of objects is at most $nt=\Oh(n)$.
Thus there exists a vertical line $\ell_i$ with $\lceil (a-1)/3 \rceil \leq i \leq \lfloor 2(a-1)/3 \rfloor$ that intersects at most $nt/(a/3+1)=\Oh(n^{2/3})$ convex hulls.

Consider now the set $V_1$ of objects to the left of $\ell_i$ whose convex hulls are disjoint from $\ell_i$; we call the intersection graph of these objects the ``left'' instance.
Similarly, the set $V_2$ of objects to the right of $\ell_i$ form the ``right'' instance.
Notice that the distance between $\ell_1$ and $\ell_a$, i.e., the width of $\bb(V)$, is at least $(a-1)/\sqrt{2}$, while the width of $\bb(V_1)$ is at most $i/\sqrt{2} \leq \frac{2}{3}(a-1)/\sqrt{2}$.
At the same time, the height of $\bb(V_1)$ is less or equal to the height of $\bb(V)$.
Consequently, $\area(V_1) \leq \frac{2}{3} \area(V)$.
Applying an analogous reasoning to $V_2$, we obtain that $\area(V_2) \leq \frac{2}{3} \area(V)$.
\end{proof}

\paragraph*{Solving small-area instances.}
Let us introduce an auxiliary problem.
The $\lhomorc{H}$ problem is a restriction of $\lhomo{H}$, where for every instance $(G,L)$, and for every $v \in V(G)$ the set $L(v)$ induced a reflexive clique in $H$.
Note that in this problem we can always focus on the subgraph induced by reflexive vertices of $H$, as irreflexive vertices do not appear in any lists.
Thus, $\lhomorc{H}$ is equivalent to $\lhomorc{H_R}$, where $H_R := H[R(H)]$.

For an instance $(G,L)$ of $\lhomo{H}$, we say that a family $\cX$ of instances of $\lhomo{H}$ is \emph{equivalent} to $(G,L)$ if the following holds: $(G,L)$ is a yes-instance if and only if $\cX$ contains at least one yes-instance.

\begin{lemma} \label{lem:reducetorc}
Let $V$ be a set of $n$ similarly-sized fat objects in $\R^2$.
Let $(G,L)$ be an instance of $\lhomo{H}$, where $G= \IG(V)$.
Then in time $n^{\Oh(\area(V))}$ we can build a family $\cY$ of instances of $\lhomorc{H}$, such that:
\begin{itemize}
\item $|\cY| = n^{\Oh(\area(V))}$,
\item each instance in $\cY$ is an induced subgraph of $G$,
\item $\cY$ is equivalent to $(G,L)$.
\end{itemize} 
\end{lemma}
\begin{proof}
Recall that $V$ can be partitioned into $\Oh(\area(V))$ cell-cliques, and consider one such cell-clique $C$.
By \cref{lem:mapclique} at most $|I(H)|$ vertices from $C$ receive colors from $I(H)$ and the remaining vertices of $C$ must be mapped to some reflexive clique of $H$.
We guess the vertices mapped to $I(H)$ along with their colors and the reflexive clique to which the remaining vertices are mapped.
As $H$ is a constant, the total number of branches created for $C$ is $|C|^{\Oh(|V(H)|)}=n^{\Oh(1)}$.
Repeating this for every clique, we result in $n^{\Oh(\area(V))}$ branches.

Consider one such a branch.
For each vertex $v$ whose color was guessed (i.e., this color is in $I(H)$), we update the lists of neighbors of $v$.
More precisely, if the color guessed for $v$ is $a$,
then we remove every nonneighbor of $a$ from the lists of all neighbors of $v$.
After that we remove $v$ from the graph.
Similarly, we update the lists of vertices $v$ that are supposed to be mapped to vertices of $R(H)$:
we remove from $L(v)$ every vertex that does not belong to the guessed reflexive clique.

Note that this way we obtained an instance of $\lhomorc{H}$, where the instance graph is an induced subgraph of $G$.
We include such an instance into $\cY$.

As the number of branches is $n^{\Oh(\area(V))}$, we obtain that $|\cY|=n^{\Oh(\area(V))}$.
Furthermore, from the way how $\cY$ was constructed, it is clear that $\cY$ is equivalent to $(G,L)$.
\end{proof}

\paragraph*{Wrapping up the proof.}
Pipelining \cref{lem:trueseparator} with \cref{lem:reducetorc} we obtain the following.

\begin{lemma}\label{lem:wrappedfatssized}
Let $H$ be a fixed graph.
Suppose that $\lhomorc{H}$ can be solved in time $2^{\Oh(n^{2/3} \log n)}$ in $n$-vertex intersection graphs of fat, similarly-sized objects, given along with a geometric representation.

\noindent Then $\lhomo{H}$ can be solved in time $2^{\Oh(n^{2/3} \log n)}$ in $n$-vertex intersection graphs of fat, similarly-sized objects, given along with a geometric representation.
\end{lemma}
\begin{proof}
Let $V$ be a set of $n$ fat, similarly-sized objects in $\R^2$ and let $G=\IG(V)$.
Let $(G,L)$ be an instance of $\lhomo{H}$.
Notice that if $G$ is disconnected, then we can solve the problem for each connected component separately.
Thus let us assume that $G$ is connected. 
We do induction on $\area(V)$; by \cref{obs:connectedarea} we have $\area(V) = \Oh(n^2)$.

If $\area(V) = \Oh(n^{2/3})$ (the actual constant in $\Oh(\cdot)$ is the constant from \cref{lem:trueseparator}),
we call \cref{lem:reducetorc} to obtain a family $\cY$ of instances of $\lhomorc{H}$, such that $|\cY|=n^{\Oh(n^{2/3})}$.
Each instance in $\cY$ is an induced subgraph of $G$, and $\cY$ is equivalent to $(G,L)$.
By our assumption, each instance in $\cY$ can be solved in time $2^{\Oh(n^{2/3} \log n)}$, and thus we can solve the problem in time
\[
2^{\Oh(n^{2/3} \log n)} + 2^{\Oh(n^{2/3} \log n)} \cdot  2^{\Oh(n^{2/3} \log n)} = 2^{\Oh(n^{2/3} \log n)},
\]
as claimed.

In the other case, we apply \cref{lem:trueseparator}, let $\ell$ be the obtained separating line.
Let $S \subseteq V$ be the set of objects whose convex hull intersects $\ell$;
by \cref{lem:trueseparator} the size of $S$ is $\Oh(n^{2/3})$.

Let $V_1,V_2$ be the partition of $V \setminus S$ into instances of each side of $\ell$, as in \cref{lem:trueseparator}.
Recall that $\area(V_1) \leq \frac{3}{4} \area(V)$ and $\area(V_2) \leq \frac{3}{4} \area(V)$.

We exhaustively guess the coloring of $S$, this results in $|V(H)|^{|S|} = 2^{\Oh(n^{2/3})}$ branches.
For each such branch we update the lists of neighbors of vertices whose color was guessed.
Now observe that the subinstances induced by $V_1$ and $V_2$ can be solved independently.
Our initial instance is a yes-instance if and only if for some guess both subinstances are yes-instances.

Denoting by $\mu$ the measure of our instance, i.e., $\area(V)$, we obtain the following recursion for the running time.
\[
F(\mu) \leq 2^{\Oh(n^{2/3})} \cdot F\left( \frac{3}{4} \mu \right),
\]
which solves to $F(\mu) \leq 2^{\Oh(n^{2/3} \log \mu)}$.
As $\mu = \Oh(n^2)$, we conclude that the total running time is $2^{\Oh(n^{2/3} \log n)}$.
\end{proof}

Before we proceed to the proof of \cref{thm:fatssized}, let us recall the following classic result by Edwards~\cite{Ed86}.

\begin{theorem}[Edwards~\cite{Ed86}]\label{thm:edwards}
For every graph $H$, every instance $(G,L)$ of $\lhomo{H}$, where every list is of size at most 2,
can be solved in polynomial time.
\end{theorem}

Now, combining \cref{lem:wrappedfatssized} with \cref{thm:edwards}, we obtain \cref{thm:fatssized}.
\thmfatssized*
\begin{proof}
Observe that in every instance of $\lhomorc{H}$, each list is of size at most 2
and thus every such instance can be solved in polynomial time by \cref{thm:edwards}.
So the result follows by \cref{lem:wrappedfatssized}.
\end{proof}

Let us mention one more family of graphs $H$, where $\lhomorc{H}$ is polynomial-time solvable. 
Feder and Hell~\cite{FEDER1998236} studied a variant of $\lhomo{H}$ called \textsc{CL-LHom}($H$), where each list is restricted to form a connected subset of $H$ (CL stands for ``connected list'').
They proved that if $H$ is reflexive, then the above problem is polynomial-time solvable for chordal graphs $H$, and \NP-complete otherwise.
We observe that for reflexive graphs $H$, the $\lhomorc{H}$ problem is a restriction of \textsc{CL-LHom}($H$) and thus algorithmic results for \textsc{CL-LHom}($H$) carry over to $\lhomorc{H}$.
Consequently, by \cref{lem:wrappedfatssized}, we obtain the following corollary.

\begin{corollary}
Let $H$ be a graph such that $H_R$ is chordal. Then $\lhomo{H}$ can be solved in time
$2^{\Oh( n ^{2/3} \log n)}$ in $n$-vertex intersection graphs of fat, similarly-sized objects,
provided that the instance graph is given along with a geometric representation.
\end{corollary}

\newpage
\section{Lower bounds for intersection graphs of fat objects}\label{sec:lower}
In this section we aim to show that the assumptions of \cref{thm:cliquebased} and \cref{thm:fatssized} cannot be dropped or significantly relaxed, by exhibiting the corresponding lower bounds.

\subsection{Fat, convex objects}
First, we show that the assumption of \cref{thm:fatssized} that the given geometric representation of the input graph consists of similarly-sized objects cannot be dropped. Thus we consider intersection graphs of convex, fat objects, but we do not assume that they are similarly sized.

\begin{figure}[htb]
\centering
\includegraphics[scale=1,page=1]{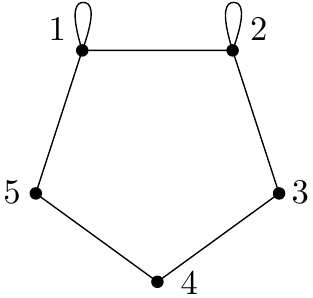}
\caption{The graph $H$ from \cref{thm:convexfat-lower}.} \label{fig:convexfat-lowerH}
\end{figure}

\begin{restatable}{theorem}{thmconvexfatlower}
%\begin{theorem}
\label{thm:convexfat-lower}
Assume the ETH.
There is a graph $H$ with $\mrc{H}=2$, such that $\lhomo{H}$ cannot be solved in time $2^{o(n)}$ in $n$-vertex intersection graphs of equilateral triangles, even if they are given along with a geometric representation.
%\end{theorem}
\end{restatable}
\begin{proof}
Let $H$ be the five-vertex cycle where exactly two adjacent vertices have loops. Clearly $\mrc{H} = 2$. We use vertex names from~\cref{fig:convexfat-lowerH}.

We reduce from 3-\textsc{Sat}. Let $\Phi$ be an instance with $N$ variables and $M$ clauses, each of which contains exactly three variables.
The ETH implies that there is no algorithm solving every such instance in time $2^{o(N+M)}$.
Denote the variables of $\Phi$ by $v_1,v_2,\ldots,v_N$; we will assume that this set is ordered.

Each variable $v_i$ is represented by a 7-vertex variable gadget depicted in \cref{fig:convexfat-gadgets}~(i). We use the notation from the figure.
It is straightforward to verify that in every list homomorphism to $H$, the triangle $x_i$ receives a different color than $y_i$.
We interpret coloring $x_i$ with the color 5 as setting the variable $v_i$ true, and coloring $x_i$ with the color 3 as setting $v_i$ false.
Under this interpretation, $y_i$ represents the value of $\lnot v_i$.

\begin{figure}[h]
\centering
\raisebox{0.8\height}{
\includegraphics[scale=1,page=2]{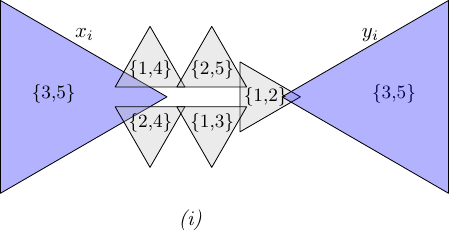}}\hspace{30pt}
\includegraphics[scale=1,page=3]{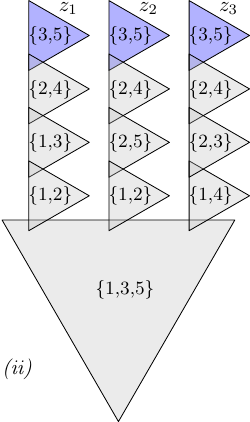}
\caption{ (i) The variable gadget and (ii) the clause gadget from the proof of \cref{thm:convexfat-lower}. Sets indicate the lists.
Triangles marked blue intersect triangles outside the gadget.} \label{fig:convexfat-gadgets}
\end{figure}

For each clause we introduce a 13-vertex clause gadget, depicted in \cref{fig:convexfat-gadgets}~(ii).
Again, it is straightforward to verify that the gadget admits a list homomorphism to $H$ if and only if at least one of vertices $z_1,z_2,z_3$ is colored 3.

The overall arrangement of variable and clause gadgets is depicted in \cref{fig:convexfat-all}. 
So far all introduced triangles are of bounded size.

\begin{figure}[htb]
\centering
\includegraphics[scale=1,page=4]{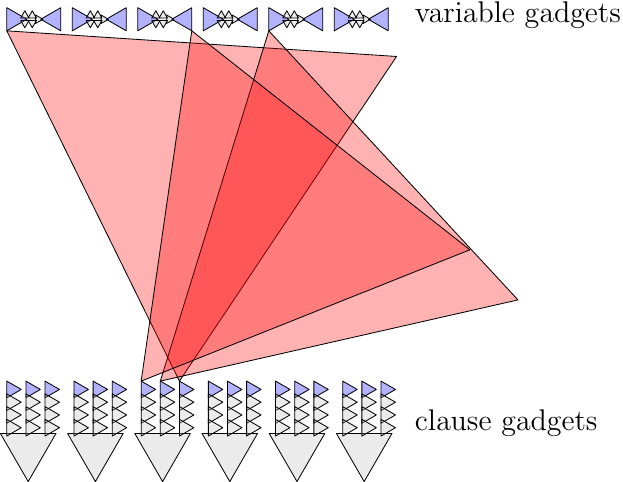}
\caption{The overall arrangement of gadgets in the proof of \cref{thm:convexfat-lower}.
Three triangles $q_{c,1}, q_{c,2}, q_{c,3}$ for a clause $c$ are depicted red.
} \label{fig:convexfat-all}
\end{figure}

The only thing left is to connect variable gadgets with clause gadgets.
Fix a clause $c = (v'_1,v'_2,v'_3)$, where the ordering of literals corresponds to the ordering of variables.
Thus there is a function $\sigma : [3] \to [N]$, such that $v'_i$ is ether $v_{\sigma(i)}$ or $\lnot v_{\sigma(i)}$.
We consider the clause gadget corresponding to $c$.
For each $i \in [3]$, we introduce an equilateral triangle $q_{c,i}$ intersecting $z_i$ and $x_{\sigma(i)}$ (if $v'_i=v_{\sigma(i)}$)
or $y_{\sigma(i)}$ (if $v'_i=\lnot v_{\sigma(i)}$), and no other triangles from vertex and clause gadgets.
This can be done if the triangle $q_{c,i}$ is much larger than the triangles from the gadgets (their diameter depends on $N$ and $M$), see \cref{fig:convexfat-all}.
The list of $q_{c,i}$ is $\{1,2\}$.
Note that this ensures that the color of $z_i$ is the same as the color of the triangle in the vertex gadget intersecting $q_{c,i}$.
Thus, by the properties of the gadgets, we observe that the constructed intersection graph admits a list homomorphism to $H$ if and only if $\Phi$ is satisfiable.

As the total number of vertices in the constructed graph is $7N + 13M + 3M = \Oh(N+M)$, the ETH lower bound follows.
\end{proof}

\subsection{Fat, similarly-sized objects}
Now, instead of focusing on restrictions on the input graph, we focus on the restrictions imposed on the target graph $H$.
We show that the assumption of \cref{thm:fatssized} that $\mrc{H} \leq 2$ cannot be relaxed.

\begin{figure}[htb]
\centering
\includegraphics[scale=1,page=5]{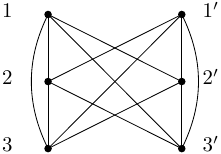}
\caption{The graph $H$ from \cref{thm:fatssized-lower}.
All vertices are reflexive; loops are not drawn to increase readability.
} \label{fig:fatssized-lowerH}
\end{figure}

\begin{restatable}{theorem}{thmfatssizedlower}
%\begin{theorem}
\label{thm:fatssized-lower}
Assume the ETH.
There is a graph $H$ with $\mrc{H}=3$, such that $\lhomo{H}$ cannot be solved in time $2^{o(n)}$ in $n$-vertex intersection graphs of fat similarly-sized triangles, even if they are given along with a geometric representation.
%\end{theorem}
\end{restatable}
\begin{proof}
Let $H$ be the graph depicted in~\cref{fig:fatssized-lowerH}; we use the vertex names from the Figure.
%Our reduction works for both $H_1$ and $H_2$, as we will never make use of the (non-)existence of the edge $1'3'$.

We reduce from $\textsc{Not-All-Equal-3-Sat}$. Let $\Phi$ be a formula with $N$ variables $v_1,v_2,\ldots,v_N$ and $M$ clauses, each of which contains precisely three nonnegated variables. The ETH implies that there is no algorithm solving every such instance in time $2^{o(N+M)}$.

Each variable $v_i$ is represented by the variable gadget depicted in \cref{fig:fatssized-gadgets}~(i).
It consists of 13 triangles and six of them (i.e., $x_1,x_2,x_3,y_1,y_2,y_3$, marked in color in \cref{fig:fatssized-gadgets}~(i)) will play a special role. 
It is straightforward to verify that the variable gadget has exactly two list homomorphisms to $H$:
\begin{enumerate}
\item $f_0$, such that $f_0(x_1)=f_0(y_1)=2$ and $f_0(x_2)=f_0(y_2)=1$, and $f_0(x_3)=f_0(y_3)=3$,
\item $f_1$, such that $f_1(x_1)=f_1(y_1)=3'$ and $f_1(x_2)=f_1(y_2)=2'$, and $f_1(x_3)=f_1(y_3)=1'$.
\end{enumerate}
We will interpret the coloring $f_j$ of the variable gadget as assigning the value $j$ to $v_i$.

The variable gadgets corresponding to distinct variables are ``stacked'' on each other, so that the corresponding triangles from different gadgets form a clique, and the triangles from different cliques intersect each other if and only if they belong to the same variable gadget, see \cref{fig:fatssized-gadgets}~(ii). Note that the corner of each special triangle in each variable gadget that points toward the center of the gadget is not covered by other triangles.

\begin{figure}[htb]
\centering
\includegraphics[scale=1,page=6]{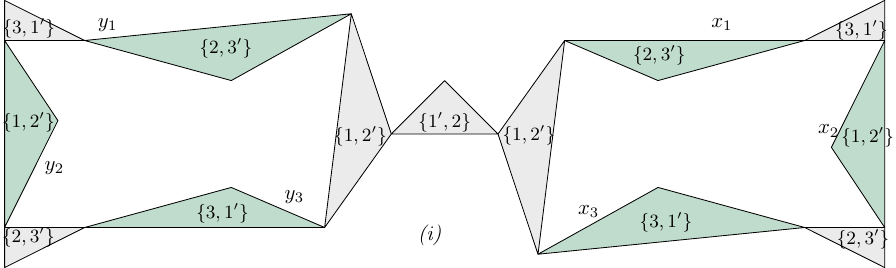}
\vskip 10pt
\includegraphics[scale=1,page=7]{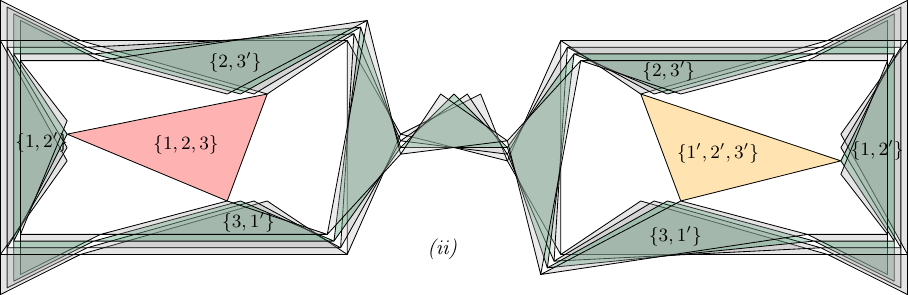}
\caption{(i) The variable gadget in the proof of \cref{thm:fatssized-lower}.
The sets indicate lists. 
(ii) Arrangement of all gadgets. A single variable gadget is colored green. The yellow triangle is $p_c$ and the red triangle is $q_c$, for some clause $c$.
} \label{fig:fatssized-gadgets}
\end{figure}

Now let us consider a clause $c=(v_i,v_j,v_k)$, such that $i < j < k$.
We introduce a triangle $p_c$ intersecting the triangle $x_1$ from the vertex gadget corresponding to $v_i$, 
the triangle $x_2$ from the vertex gadget corresponding to $v_j$, and the triangle $x_3$ from the vertex gadget corresponding to $v_k$.
Similarly, we introduce a triangle $q_c$ intersecting the triangle $y_1$ from the vertex gadget corresponding to $v_i$, 
the triangle $y_2$ from the vertex gadget corresponding to $v_j$, and the triangle $y_3$ from the vertex gadget corresponding to $v_k$, see \cref{fig:fatssized-gadgets}~(ii).
The list of $p_c$ is $\{1',2',3'\}$, and the list of $q_c$ is $\{1,2,3\}$.

Note that we can choose a color for $p_c$ if and only if for at least one variable of $c$, its variable gadget gets colored according to $f_1$ (i.e., this variable is set true).
Similarly, we can choose a color for $q_c$ if and only if for at least one variable of $c$, its variable gadget gets colored according to $f_0$ (i.e., this variable is set false). 
Consequently, both triangles corresponding to $c$ can be colored if and only if the clause $c$ is satisfied.

Note that all triangles are similarly-sized, and the vertex set of the constructed graph can be covered with 15 cliques. Furthermore for each of the cliques, the vertices of $H$ appearing in the lists form a reflexive clique in $H$. Thus the constructed graph can be seen as an instance of $\lhomorc{H}$.

The total number of triangles is $13N + 2M = \Oh(N+M)$, so the ETH lower bound follows.
\end{proof}

%\begin{remark}
%Let us point out that actually graphs in \cref{fig:fatssized-lowerH} are the only graphs $H$ with at most 6 vertices for which $\lhomorc{H}$ is
%\NP-hard. This was verified with quite tedious case analysis using the list of all connected graphs on at most 6 vertices~\cite{DBLP:journals/dm/CvetkovicP84} and the observations about the complexity of $\lhomorc{H}$ that we made in \cref{sec:algo-fatssized}.
%\end{remark}

\subsection{Disks}
In this section we show that the assumption that $\mrc{H} \leq 1$ in \cref{thm:cliquebased} cannot be significantly improved.
Our goal is to prove the following theorem.

\thmdiskslower*

We reduce from 3-\textsc{Sat}. Let $\ell_1,\dots,\ell_t$ be the literals of the formula $\Phi$ on $N$ variables and $M$ clauses, i.e., the $i$-th clause consists of literals $\ell_{3i-2},\ell_{3i-1}$, and $\ell_{3i}$, and $t=3M$. Let $k=1 + \lceil \log t \rceil$ be the number of binary digits required to represent numbers up to $t$.

%\subparagraph*{Construction overview.}
\paragraph{Construction overview.}
The construction will have some variable gadgets placed at the top, consisting of two disks with lists $\{T,F\}$, where the value of the first disk correspond to setting the variable true or false.

The bulk of the construction will consist of large cliques of disks of various sizes, and in each clique the disks will correspond to some specific subsets of literals. All of these cliques will have lists of size $2$, where the assigned colors correspond to the literal being true or false. At the top, the initial clique will have all the literals arranged by the index of the corresponding variable, i.e., starting with the positive literals of $x_1$, then the negative literals of $x_1$, then the positive literals of $x_2$, etc.

Suppose now that each literal index $i$ is represented with a binary number of $t$ digits (with leading zeros as necessary).
Then we will use a so-called \emph{divider gadget} to partition the set of literals to two subsets: the first subset will contain disks for those literals $\ell_i$ where the first binary digit of $i$ is $0$ , and the other subset will contain those where the first binary digit of $i$ is $1$. Using then two smaller copies of the divider gadget, we further partition both sets according to the second, third, etc. binary digits, creating a structure resembling a binary tree of depth $\log N + \Oh(1)$. At the leaves, the cliques contain a single disk, and the leaves are ordered in increasing order of the index $i$, that is, the literals of each clause $c_j$ appear at three consecutive leaves. We can then use a clause gadget for each consecutive triplet to check the clauses.

\begin{figure}[t]
\centering
\includegraphics[page=3,width=\textwidth]{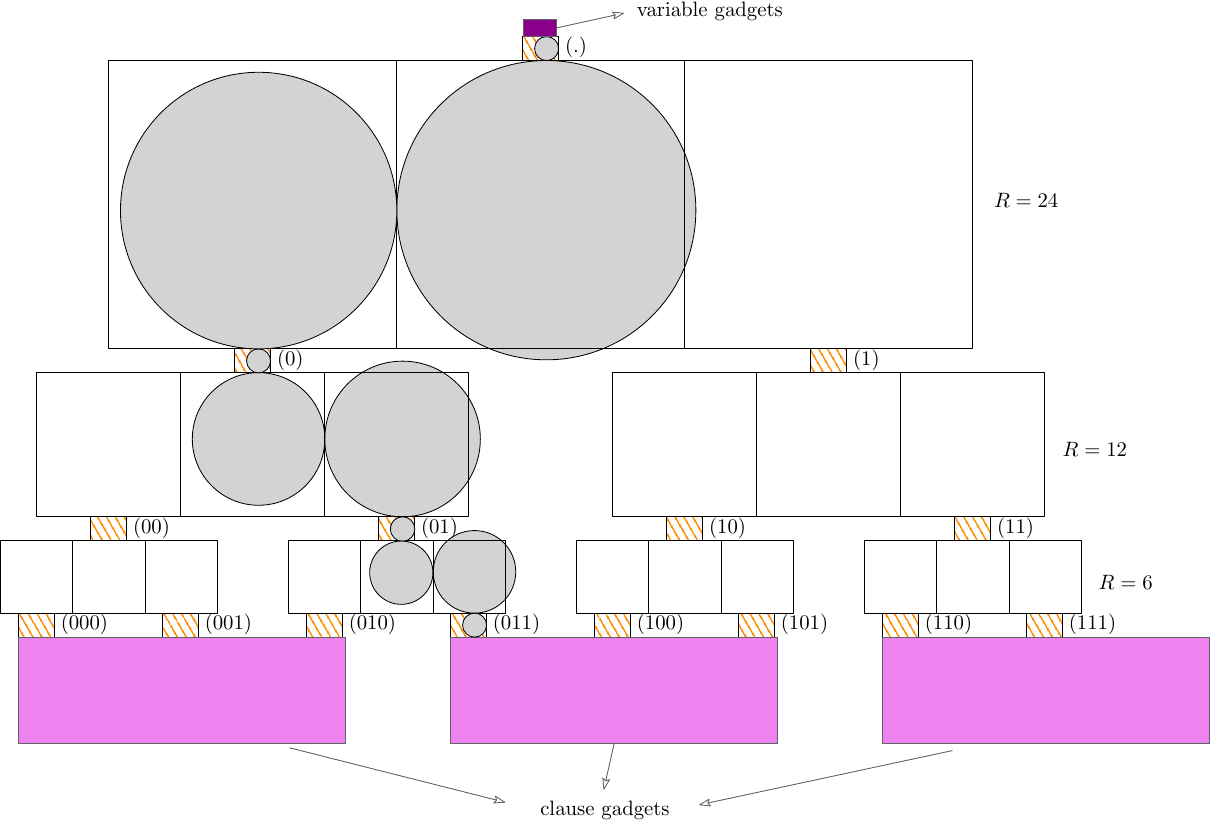}
\caption{Overview of the construction, with the path of the literal with binary index `$011$' highlighted. The lined rectangles are literal cliques, with the common prefix of the literal indices. The triplets of squares represent subset turning and divider gadgets.}\label{fig:disk_overview}
\end{figure}

We will now explain the construction in detail.

\begin{figure}[t]
\centering
\includegraphics[page=4]{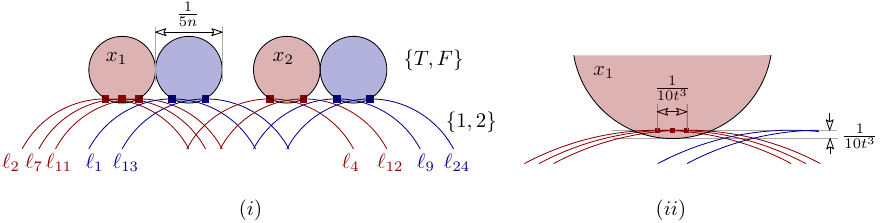}
\caption{(i) Variable gadgets with the top of the connecting literal clique. The first (red) disk of the variable gadget intersects only the disks of the positive literals of the variable, and the second (blue) disk intersects only the disks of the negative literals. (ii) A close-up of the first disk of $x_1$ with the positive intersecting literals. Note that no other literal disk intersects the red disk of $x_1$ because the negative literals and the literals of other variables are relatively far away.}\label{fig:disk_variable}
\end{figure}
\begin{figure}[h]
\centering
\includegraphics[page=5]{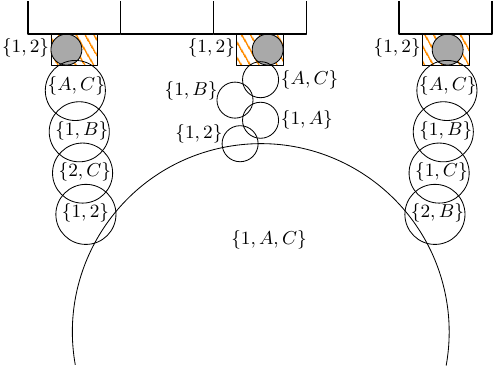}
\caption{A clause gadget attached to three consecutive ``leaves''.}\label{fig:disk_clause}
\end{figure}

%\subparagraph*{Literal cliques, variable and clause gadgets.}
\paragraph{Literal cliques, variable and clause gadgets.}
A literal clique consists of at most $t$ disks of unit radius that are on the same horizontal line. Each literal clique will contain the set of literals whose index starts with some fixed binary prefix $s$ of length at most $k$, and these cliques will be connected by other gadgets, creating a binary tree. Let us denote the set of literals with prefix $s$ by $L_s$. The initial literal clique will have disks for all the literals (the set $L_\emptyset$). In the initial clique, these literals will be ordered from left to right according to the corresponding variable and their sign; more precisely, we first have the disks of the literals where variable $x_1$ is positive, then the literals where $x_1$ is negative. This is then followed by the positive and then the negative literals of $x_2$, etc. The centers of the disks where $x_i$ is positive are placed at equal distances from each other, within an interval of length $\frac{1}{10t^3}$ on the $x$-axis. We then translate this interval to the right by $\frac{1}{5N}$, and place the centers of the disks where $x_i$ is negative. Translating the interval further to the right by $\frac{1}{4N}$ we can place the positive literals of $x_{i+1}$. Notice that at the end of the process, all disk centers are within horizontal distance of at most $N\cdot \frac{1}{5N}+ (N-1)\cdot \frac{1}{4N}+\frac{1}{10t^3}<1$ from each other, i.e., these disks form a clique in the intersection graph.

Each later clique will contain the subset $L_s$ of literals, positioned the same way, just translated somewhere else in the plane. Note that for prefixes $s$ of length $k$, the set $L_s$ is a singleton, it contains the literal of binary index $s$. These literal cliques will correspond to the leaves of our construction. All literal cliques have lists $\{1,2\}$, corresponding to the literal being set to true or false, respectively.

The variable gadgets connect to the top of the initial literal clique that contains all literals. The gadget for $x_i$ consists of two disks of diameter $\frac{1}{5N}$ corresponding to the variable and its negation, see \cref{fig:disk_variable}. The two disks touch each other, and have their centers on the line $y=1+\frac{1}{5N}-\frac{1}{10t^3}$ so that the first disk contains the topmost points of the disks corresponding to the positive literals of $x_i$, while the second contains the topmost points of the disks corresponding to the negative literals of $x_i$. It is routine to check that each disk is only intersected by the corresponding literal disks. The disks of the variable gadget have list $\{T,F\}$, and we interpret these colors on the first disk of $x_i$ as setting $x_i$ to true or false, respectively. The colors $\{1,2,T,F\}$ form a 4-cycle in $H$ with $1,2$ being reflexive vertices, see \cref{fig:disk_H}~(i). It is routine to check that the first disk of $x_i$ has color $T$ if and only if its positive literals get color $\{1\}$ and its negative literals get color~$\{2\}$.

Our clause gadget is depicted in \cref{fig:disk_clause}. Our construction will ensure that consecutive singleton literal cliques at the leaves have a gap of length $9$ between them, therefore the centers of the disks in them have a distance between $11$ and $13$. It is easy to construct a rigid structure from disks that induces the same subgraph regardless of the exact location of each disk within its rectangle. Our clause gadget uses the same idea as already seen in \cref{fig:convexfat-lowerH}, The colors $1,2,A,B,C$ form a $5$-cycle with reflexive vertices at $1,2$. See \cref{fig:disk_H}~(i) for a picture of the relevant part of $H$. Note that the literal disks have lists $\{1,2\}$, and the first gadget disks have lists $\{A,C\}$, i.e., they are colored $C$ if and only if the corresponding literal is true. Thus the gadget has a correct coloring if and only if at least one of the three literal disks have color $1$.

%\subparagraph*{Subset turning and the divider gadget.}
\paragraph{Subset turning and the divider gadget.}
Our task now is to connect a literal clique to its children by dividing its literals into two subsets, keeping the information carried for each individual literal. First, we show how we can create a turn gadget using disks of any size.

\begin{figure}[t]
\centering
\includegraphics[page=1,width=\textwidth]{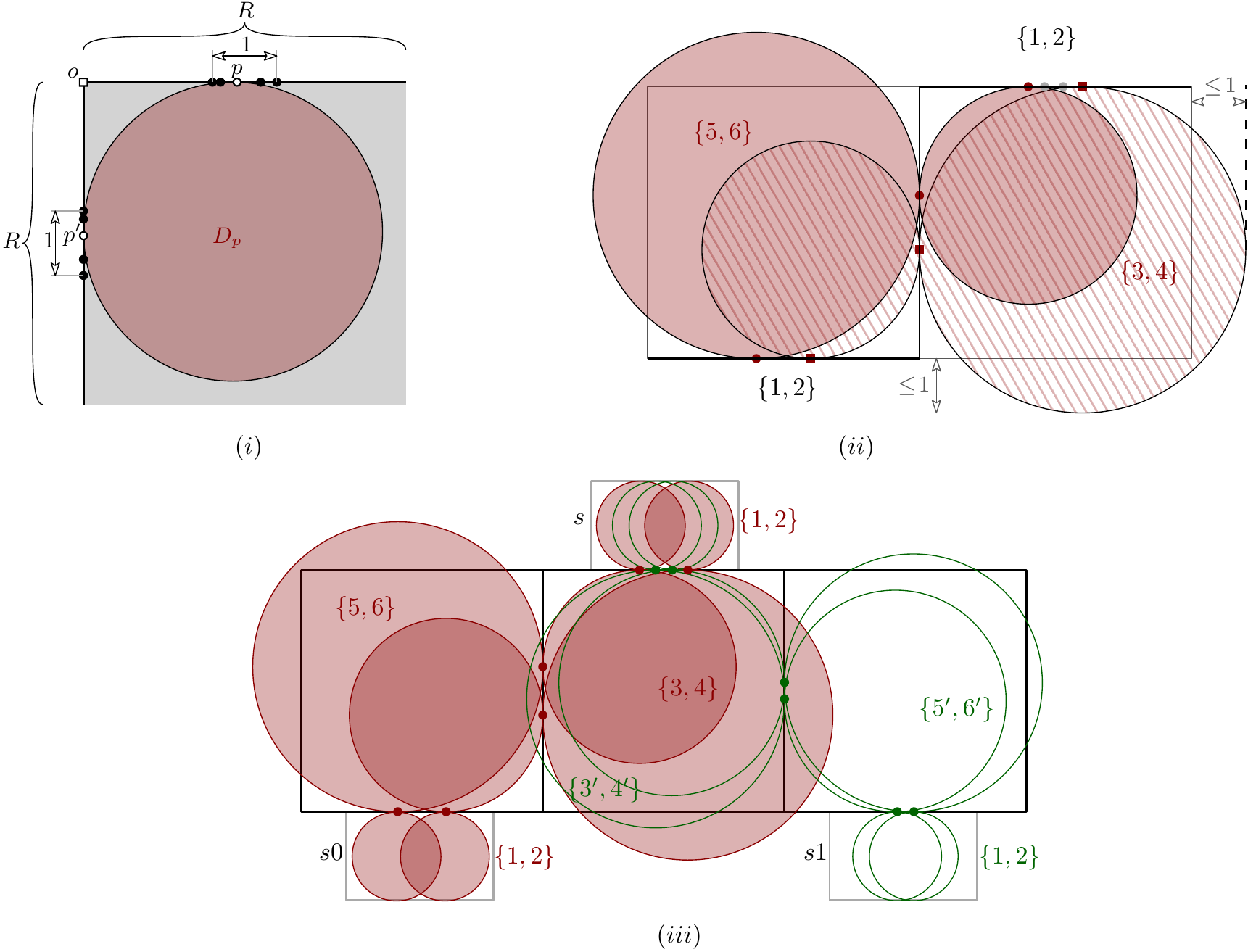}
\caption{(i) Unique disk of radius between $\left[\frac{R-1}{2},\frac{R+1}{2}\right]$ touching two perpendicular lines. (ii) Making two turns with some subset of the literals (iii) Dividing the set of literals into two arbitrary subsets (red versus green disks) using overlaid turns.}\label{fig:disk_divider}
\end{figure}

Consider a horizontal segment of length $R$ with its left endpoint at the  point $o$, and let $p$ be a point where a disk in some literal clique touches the segment from above. Suppose moreover that $p$ is somewhere in the length $1$ interval at the middle of the segment, see \cref{fig:disk_divider}~(i). Then the \emph{turning disk at $p$} is the unique disk $D_p$ that touches the segment from below and has radius $|op|$. Note that if we draw a vertical segment of length $R$ with top endpoint $o$, then it will touch the disk on the left at some point $p'$ where $|op|=|op'|$. The \emph{turning gadget} is simply a collection of turning disks for some custom set of points in the middle length-$1$ interval. We can represent the gadget with a square of side length $R$ whose top side is the initial segment. Note that some turning disks may not be completely covered by the square, but since $p$ is required to be in the middle length-1 interval, the disks can protrude at most distance $1$ beyond the boundary of the square. Also note that we can create an analogous gadget with disks that touch any pair of consecutive sides of the side-length $R$ square.

We can glue two turning gadgets together as depicted in \cref{fig:disk_divider}~(ii). The disks of the first (right) gadget have lists $\{3,4\}$, and the disks of the second (left) have lists $\{5,6\}$. In the graph $H$, we have $1,2,3,4$ as well as $3,4,5,6$ and $5,6,1,2$ form induced $4$-cycles. The connecting literal cliques have disks with lists $\{1,2\}$, and all of the colors $\{1,2,3,4,5,6\}$ are reflexive vertices of $H$, see \cref{fig:disk_H}~(ii). It is routine to check that the turning disks receive odd colors if and only if the corresponding disks in the literal cliques have color $1$.

Finally, we can overlay such a glued turning gadget with its mirror image, as depicted in \cref{fig:disk_divider}~(iii). In the mirror image, the disks of the first turning gadget get the list $\{3',4'\}$, and the disks of the second turning gadget get the list $\{5',6'\}$.
The vertices $1,2,3',4',5',6'$ induce the same graph as vertices $1,2,3,4,5,6$. These four turns together define a \emph{divider gadget} of size $R$. If the literal clique at the top contained the disks of index prefix $s$, then we use the first two turns (going to the left child, red disks in \cref{fig:disk_divider}~(iii)) only on the touching points for literals with prefix $s0$, and the other two turns (going to the right child, green disks in \cref{fig:disk_divider}~(iii)) only for the touching points for literals with prefix $s1$.

Notice however that inside the turning gadgets, there may be arbitrary intersections between red and green disks, therefore $3,4,5,6$ and $3',4',5',6'$ form a complete bipartite graph in $H$, see \cref{fig:disk_H}~(iii). Clearly the two sides of the gadget do not interfere and disks with colors $3',4',5',6'$ have an odd number if and only if the corresponding disk at the top literal clique has color $1$.

\begin{figure}
\centering
\includegraphics[page=2,width=\textwidth]{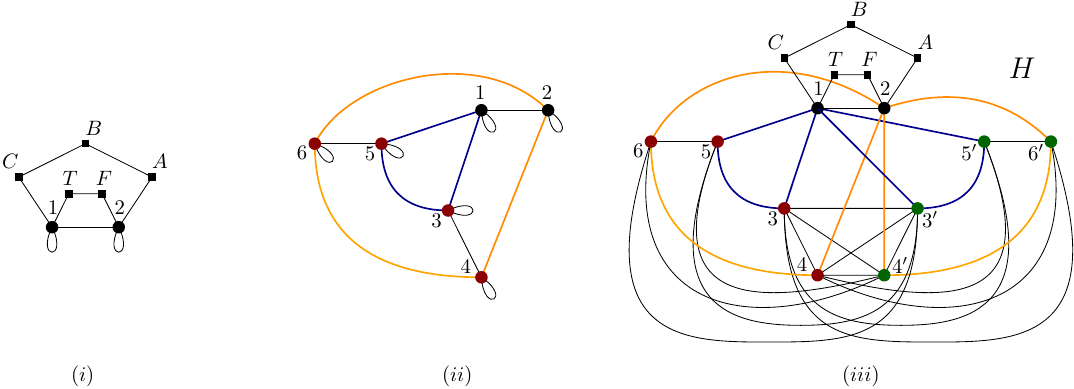}
\caption{(i) The part of $H$ used in the variable and clause gadgets. (ii) The part of $H$ responsible for propagating the truth of literals in the left side of the division. Blue edges propagate true literals, orange edges propagate a false literals. (iii) The graph $H$. Note that all numbered vertices are reflexive, and vertices with a letter are not reflexive.}\label{fig:disk_H}
\end{figure}

%\subparagraph*{The proof of Theorem~\ref{thm:disks-lower}.}
\paragraph{The proof of \cref{thm:disks-lower}.}
\begin{proof}
Recall that our formula has $t$ literals, and each literal index can be represented by a binary string of length $k=\lceil\log_2 t\rceil +1$.

We place the initial literal clique together with the variable gadgets as described in the construction. At the bottom of this literal clique the disks touch a length $1$ interval. We attach a divider gadget of size $R=6\cdot 2^{k-1}$ to this interval (see \cref{fig:disk_overview}). We then use the divider gadget to propagate the values stored in the literals to the children with prefix $s=0$ on the left and $s=1$ on the right. For literal cliques of prefix length $\len(s)$, we attach dividers of size $R=6\cdot 2^{k-1-\len(s)}$. At the bottom, we end up with singleton literal cliques hanging off of literal gadgets of size $R=6$. One can verify that the gaps between literal cliques of consecutive leaves have length $9$ (that is, the right  side of the $2\times 3$ rectangle covering the first leaf and the left side of the rectangle of the next leaf has distance $9$). It is also easy to verify that the turning disks of distinct divider gadgets are disjoint: recall that the disks protrude beyond the boundary of the base square by at most $1$, and the literal cliques have height $2$.

Based on the formula, the described set can clearly be constructed in polynomial time. Each literal has corresponding disks in $\Oh(k)=\Oh(\log t)$ gadgets, and each literal clique and divider has $\Oh(1)$ disks per represented literal. Additionally, the variable and clause gadgets have constant size. Thus for a $3$-CNF formula of $N$ variables and $M$ clauses with $t=3M$ literals, there are $\Oh(t\log t + M + N) = \Oh((M+N)\log(M + N))$ disks in the construction, which implies the desired lower bound under the ETH.
\end{proof}

\newpage
\section{Weighted generalizations of LHom(\emph{H})} \label{sec:weighted}
In this section we consider two weighted generalizations of the $\lhomo{H}$ problem, called \emph{Min Cost Homomorphism}~\cite{DBLP:journals/ejc/GutinHRY08} and the \emph{Weighted Homomorphism}~\cite{DBLP:journals/jcss/OkrasaR20}.
We denote them, respectively, by $\mchomo{H}$ and $\whomo{H}$.

\subsection{{Min Cost Homomorphism}}
For a fixed graph $H$, the instance of $\mchomo{H}$ is a graph $G$ equipped with a weight function $\wei : V(G) \times V(H) \to \QQ_{\geq 0}$, and an integer $k$.
The value of $\wei(v,a)$ is interpreted as a cost of assigning the color $a$ to the vertex $v$.
The cost of a homomorphism $f : G \to H$ is defined as $\sum_{v \in V(G)} \wei(v,f(v))$.
The problem asks whether $G$ admits a homomorphism to $H$ with total cost at most $k$.

Note that $\mchomo{H}$ is indeed a generalization of $\lhomo{H}$.
For an instance $(G,L)$ of $\lhomo{H}$ we can construct an equivalent instance $(G,\wei,0)$ of $\mchomo{H}$
by setting $\wei(v,a)=0$ if $a \in L(v)$, and $\wei(v,a)=1$ if $a \notin L(v)$.

However, the $\mchomo{H}$ is more robust, as in addition to \emph{hard} constraints (edges of $H$) it allows to express \emph{soft} constraints (weights). 
The most prominent special case is when $H$ is the graph depicted in \cref{fig:vertexcover}.
It is straightforward to observe that if $\wei(v,1) = 0$ and $\wei(v,2)=1$ for every vertex $v$ of the instance graph $G$,
then $G$ admits a homomorphism to $H$ with total cost at most $k$ if and only if it admits a vertex cover of size at most $k$ (or, equivalently, an independent set of size at least $|V(G)|-k$).

\begin{figure}[h]
\centering
\includegraphics[scale=1,page=8]{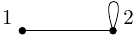}
\caption{The graph $H$, for which $\mchomo{H}$ is equivalent to \textsc{Min Vertex Cover}.} \label{fig:vertexcover}
\end{figure}

It is straightforward to observe that the divide-\&-conquer approach used in \cref{thm:cliquebased-general} can be generalized to the weighted setting.
Thus we immediately obtain the following strengthening of \cref{thm:cliquebased}.

\begin{theorem}\label{thm:algomchomo}
Let $H$ be a graph with $\mrc{H} \leq 1$. Then $\mchomo{H}$ can be solved in time:
\begin{enumerate}[(a)]
\item $2^{\Oh(\sqrt{n})}$ in $n$-vertex intersection graphs of fat convex objects,
\item $2^{\Oh(n^{2/3}\log{n})}$ in $n$-vertex pseudodisk intersection graphs.
\end{enumerate}
provided that the instance graph is given along with a geometric representation.
\end{theorem}

Now let us consider the possibility of generalizing \cref{thm:fatssized} in an analogous way.
Note that the first two steps, i.e., applying \cref{lem:trueseparator} and \cref{lem:reducetorc}, carry over to the more general setting.
However, recall that the last step in the proof of \cref{thm:fatssized} was the application of \cref{thm:edwards}, which solves every instance of $\lhomo{H}$, where each list is of size at most 2, by a reduction to \textsc{2-Sat}.
It is well-known that the weighted variant of \textsc{2-Sat} if \NP-complete and cannot be solved in subexponential time, unless the ETH fails~\cite{DBLP:journals/amai/Porschen07}. However, this does not rule out the possibility of dealing with the last step in some other way.

In the next theorem we show that this is not possible, and the complexity of $\lhomo{H}$ and $\mchomo{H}$, when $\mrc{H}=2$, differs in intersection graphs of fat-similarly sized objects.

\begin{restatable}{theorem}{thmmchomolower}
%\begin{theorem}
\label{thm:mchomo-lower}
Assume the ETH.
There is a graph $H$ with $\mrc{H}=2$, such that $\mchomo{H}$ cannot be solved in time $2^{o(n)}$ in $n$-vertex intersection graphs of fat smilarly-sized triangles, even if they are given along with a geometric representation.
%\end{theorem}
\end{restatable}

\begin{proof}
Let $H$ be the reflexive $C_4$ with consecutive vertices 1, 2, 3, 4.

We reduce from $\textsc{Min Vertex Cover}$, let $(G',k')$ be an instance with $N$ vertices $v_1,v_2,\ldots,v_N$ and $M$ edges.
It is well-known that the existence of an algorithm solving every such instance in time $2^{o(N+M)}$ would contradict the ETH~\cite{DBLP:books/sp/CyganFKLMPPS15}.

We will construct an equivalent instance $(G,\wei,k)$ of $\mchomo{H}$.
To simplify the description, we will also prescribe lists $L : V(G) \to 2^{\{1,2,3,4\}}$. Note that this does not really change the problem,
as the fact that some $a$ is not in the list of some vertex $x \in V(G)$ can be expressed by setting $\wei(x,a)$ to some large value (larger that our budget $k$). 

Each vertex $v_i \in V(G')$ is represented by the vertex gadget depicted in \cref{fig:mchomo-gadgets}~(i). It consists of four triangles, two of each (denoted by $x$ and $y$ in the figure) will interact with the other gadgets.
It is straightforward to observe that the vertex gadget admits exactly two homomorphisms to $H$ that satisfy lists:
\begin{enumerate}
\item $f_0$, such that $f_0(x) = 1$ and $f_0(y)=4$, and
\item $f_1$, such that $f_1(x) = 2$ and $f_0(y)=3$.
\end{enumerate}
We will interpret the coloring $f_0$ as not selecting $v_i$ to the vertex cover, and $f_1$ as selecting it.
Thus we set $\wei(x,1)=1$ and the remaining weights are set to 0 (recall that in the final step we will modify the weights to get rid of lists, this step is not included in the definition of the weights above).

The vertex gadgets corresponding to distinct vertices of $G'$ are ``stacked'' on each other, so that the corresponding triangles from different  gadgets form a clique, and the triangles from different cliques intersect each other if and only if they belong to the same gadget, see \cref{fig:mchomo-gadgets}~(ii). Note that the corner of each special triangle in each vertex gadget that points toward the center of the gadget is not covered by other triangles.

Now we need to ensure that for each edge at least one of its endvertices must be selected to the vertex cover.
Consider an edge $v_iv_j$, where $i<j$.
We introduce a triangle $z$ intersecting the vertex $x$ from the vertex gadget corresponding to $v_i$ and the vertex $y$ from the vertex gadget corresponding to $v_j$ (see \cref{fig:mchomo-gadgets}~(ii)) and no other triangles from vertex gadgets.
The list of $z$ is $\{2,3\}$. It is straightforward to verify that $z$ can be colored if and only if at least one of the vertex gadgets corresponding to $v_i$ and $v_j$ is colored according to the coloring $f_1$.

Consequently, $G$ admits a (list) homomorphism to $H$ with cost at most $k'$ if and only if $G'$ has a vertex cover of size at most $k'$.

Note that all triangles are similarly-sized, and the vertex set of the constructed graph can be covered with five cliques. Furthermore for each of the cliques, the vertices of $H$ appearing in the lists form a reflexive clique in $H$.
The total number of triangles is $4N + M = \Oh(N+M)$, so the ETH lower bound follows.
\end{proof}

\begin{figure}[htb]
\centering
\includegraphics[scale=1,page=9]{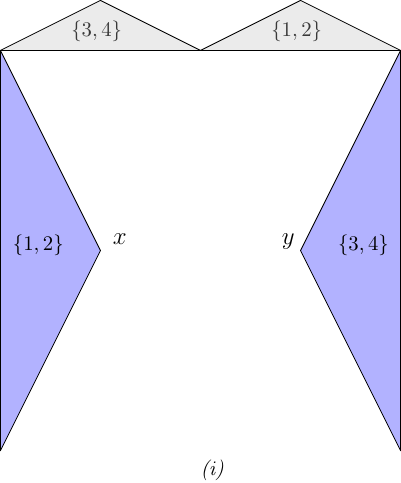}%
\hskip 20pt
\includegraphics[scale=1,page=10]{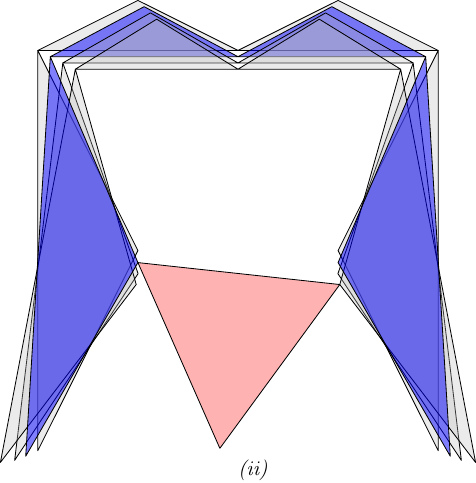}
\caption{(i) The vertex gadget in the proof of \cref{thm:mchomo-lower}.
The sets indicate lists. 
(ii) Arrangement of all gadgets. A single vertex gadget is colored blue. The red triangle represents an edge of $G'$.
} \label{fig:mchomo-gadgets}
\end{figure}

\subsection{{Weighted Homomorphism}} \label{sec:whomo}
In the $\whomo{H}$ problem in addition to vertex weights we also have edge weights.
An instance of $\whomo{H}$ is a triple $(G,\wei,k)$ where $G$ is a graph, $\wei : V(G) \times V(H) \cup E(G) \times E(H) \to \QQ_{\geq 0}$, and $k$ is an integer.
The \emph{weight} of a homomorphism $f : G \to H$ is defined as $\sum_{v \in V(G)} \wei(v,f(v)) + \sum_{uv \in E(G)} \wei(uv,f(u)f(v))$.
Again we ask for the existence of a homomorphism of weight at most $k$.

Edge weights allow us to express even more problem. For example, consider the graph $H$ from \cref{fig:maxcut}.
Let $(G,\wei,|E(G)|-k)$ be an instance of $\whomo{H}$, where all vertex weights are 0, and for every edge $e \in E(G)$ the weights are $\wei(e,12)=0$ and $\wei(e,11)=\wei(22)=1$.
It is straightforward to observe that a homomorphism of weight at most $|E(G)|-k$ corresponds to a (non-necessarily induced) bipartite subgraph of $G$ with at least $k$ edges. This is equivalent to the classic \textsc{Simple Max Cut} problem.
By adjusting edge weights, we can also encode the \textsc{Weighted Max Cut} problem, where every edge can have its own weight and we look for a cut of maximum weight.

\begin{figure}[h]
\centering
\includegraphics[scale=1,page=11]{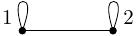}
\caption{The graph $H$, for which $\whomo{H}$ is equivalent to \textsc{Max Bipartite Subgraph}.} \label{fig:maxcut}
\end{figure}

We observe that the argument in \cref{thm:algomchomo} generalizes to the $\whomo{H}$ problem.
The only difference is that after guessing the colors of the vertices in the separator $S$,
we need take care of the weights of edges joining the separator with the rest of the graph, as these edges are not present in $G-S$.
However, it is easy to transfer these edges to their endpoints in $V(G) \setminus S$ by updating the vertex weights.
Thus we obtain the following strengthening of \cref{thm:algomchomo}.

\begin{theorem}
Let $H$ be a graph with $\mrc{H} \leq 1$. Then $\whomo{H}$ can be solved in time:
\begin{enumerate}[(a)]
\item $2^{\Oh(\sqrt{n})}$ in $n$-vertex intersection graphs of fat convex objects,
\item $2^{\Oh(n^{2/3}\log{n})}$ in $n$-vertex pseudodisk intersection graphs.
\end{enumerate}
provided that the instance graph is given along with a geometric representation.
\end{theorem}

On the other hand, as $\whomo{H}$ is a generalization of $\mchomo{H}$, the hardness from \cref{thm:mchomo-lower} transfers to $\whomo{H}$.
However, in this case we can get a much stronger lower bound. Indeed, recall that whenever $\mrc{H} \geq 2$, then \textsc{Weighted Max Cut} is a special case  of $\whomo{H}$ (by adjusting vertex weights we can ensure that the only vertices that can be used are the ones that form a reflexive $K_2$). Since \textsc{Weighted Max Cut} is \NP-hard and, assuming the ETH, cannot be solved in subexponential time in \emph{complete graphs}~\cite{DBLP:journals/njc/BodlaenderJ00}, which can be realized as intersection graphs of any geometric objects, we obtain an analogous lower bound for $\whomo{H}$.

\newpage
\section{Conclusion and open problems}\label{sec:conclusion}
Let us conclude the paper with pointing out some direction for further research.

\paragraph{Algorithms for unit disk graphs.}
One of the best studied classes of intersection graphs are \emph{unit disk intersection graphs}.
They are known to admit many nice structural properties that can be exploited in the construction of algorithms~\cite{DBLP:journals/dm/ClarkCJ90,DBLP:journals/jacm/BonamyBBCGKRST21}.
However, we were not able to obtain any better results that the general ones for (pseudo)disk intersection graphs given by \cref{thm:cliquebased}.
On the other hand we were not able to show that subexponential algorithms for this class cannot exist. We believe that obtaining improved bounds for unit disk graphs is an interesting and natural problem. 

Let us point out that there are three natural places where one could try to improve our hardness reduction in \cref{thm:disks-lower}:
(a) to avoid using disks of unbounded size,
(b) to show hardness for some $H$ with $\mrc{H} \in \{2,3\}$, and
(c) to improve the lower bound to $2^{\Omega(n)}$ (instead of $2^{\Omega(n / \log n)}$).

\paragraph{Robust algorithms.}
Recall that the algorithmic results from \cref{thm:cliquebased} and \cref{thm:fatssized} necessarily require that the instance graph is given with  a geometric representation. This might be a serious drawback, as recognizing many classes of geometric intersection graphs is \NP-hard~\cite{DBLP:journals/dm/HlinenyK01} or even $\exists\R$-hard~\cite{Schaefer2017,DBLP:journals/sigact/Cardinal15}.

On the other hand, the algorithm from \cref{thm:strings} can be made \emph{robust}~\cite{DBLP:journals/jal/RaghavanS03}:
the input is just a graph, and it either returns a correct solution, or (also correctly) concludes that the instance graph is not in our class.
Such a conclusion can be reached if the exhaustive search for balanced separators of given size fails.
Note that it might happen that the instance graph is not a string graph, but still has balanced separators of the right size -- then the algorithm returns the correct solution.

We believe it is interesting to investigate if the algorithm from \cref{thm:cliquebased} and \cref{thm:fatssized} can be made robust.
In particular, can one find clique-based separators of weight $\Oh(n^\alpha)$ for $\alpha<1$ in time $2^{n^\alpha \log^{\Oh(1)}n}$?
Note that such separators might be of size linear in the number of vertices.

\paragraph{Complexity of Simple Max Cut in (unit) disk graphs.}
Recall from section \cref{sec:whomo} that \textsc{Simple Max Cut} is a special case of the Weighted Homomorphism problem.
It is known that \textsc{Simple Max Cut} is \NP-hard in unit disk graphs~\cite{DBLP:journals/tcs/DiazK07}. However, the hardness reduction introduced a quadratic blow-up in the instance size and thus only excludes a $2^{o(\sqrt{n})}$-algorithm (under the ETH).
We believe it is interesting to study whether the problem can indeed be solved in subexponential time in unit disk graphs.
Let us point out that an $2^{\Oh(\sqrt{n})}$ algorithm is known for (a slight generalization of) \emph{unit interval graphs}, which form a subclass of unit disk graphs, but have much simpler structure and it is not even known if the problem is \NP-hard in this class~\cite[Section 3.2]{DBLP:journals/algorithmica/KratochvilMN21}.

\paragraph{Improving the running time.}
Using the $\ork{3}(a,b)$-gadget from \cref{thm:gadgets} and a reduction from \textsc{Planar 3-Sat}~\cite{DBLP:journals/dam/Kratochvil94},
one can easily show that for every non-bi-arc graph $H$ that the $\lhomo{H}$ problem is \NP-hard in planar graphs,
which form a subclass of disk intersection graphs~\cite{koebe1936kontaktprobleme} and of segment graphs~\cite{ChalopinG09,DBLP:conf/soda/GoncalvesIP18}.
Thus in particular the algorithm from \cref{thm:string-non-predacious} cannot be improved to a polynomial one, unless \P=\NP.
However, the reduction above only excludes a $2^{o(\sqrt{n})}$-algorithm, under the ETH.
We actually believe that $2^{\widetilde{\Oh}(\sqrt{n})}$ should be the right complexity bound for $\lhomo{H}$ in string graphs, where $H$ is non-predacious.

\newpage
\bibliographystyle{plain}
\bibliography{main}

\end{document}